\documentclass[%
 reprint,
 superscriptaddress,
 amsmath,amssymb,
 aps,
floatfix,
]{revtex4-2}

\usepackage[utf8]{inputenc}
\usepackage[T1]{fontenc}
\usepackage{graphicx}
\usepackage[svgnames]{xcolor}
\usepackage[version=4]{mhchem}
\usepackage{physics}
\usepackage{bbold}
\usepackage{quantikz}
\usepackage{caption}
\usepackage{microtype,mathrsfs,amsthm,amsfonts,latexsym,mathtools,graphicx,enumitem,booktabs,bm,xspace,float,nicefrac,mathdots,subcaption,ellipsis,overpic,enumitem,mleftright,framed,stmaryrd,overpic,placeins,color,soul}
\usepackage[normalem]{ulem}

\usepackage[colorlinks, pagebackref=true]{hyperref}
\usepackage[capitalize,nameinlink]{cleveref}

\newtheorem{thm}{Theorem}[section]\crefname{thm}{Theorem}{Theorems}
\newtheorem*{thm*}{Theorem}
\newtheorem{lem}[thm]{Lemma}\crefname{lem}{Lemma}{Lemmas}
\crefname{prop}{Proposition}{Propositions}
\crefname{cor}{Corollary}{Corollaries}
\theoremstyle{definition}
\newtheorem{dfn}[thm]{Definition}\crefname{def}{Definition}{Definitions}
\theoremstyle{remark}

\hypersetup{
  linkcolor=[rgb]{0.3,0.3,0.6},
  citecolor=[rgb]{0.2, 0.6, 0.2},
  urlcolor=[rgb]{0.6, 0.2, 0.2}
}

\newcommand{\eps}{\varepsilon}
\newcommand{\CC}{\mathbb{C}}
\newcommand{\RR}{\mathbb{R}}
\newcommand{\NN}{\mathbb{N}}
\newcommand{\ot}{\otimes}

\newcommand{\bigO}{\mathcal O}
\newcommand{\gap}{\omega}
\newcommand{\cH}{\mathcal H}
\newcommand{\act}{A}
\newcommand{\env}{C}
\newcommand{\bath}{B}

\newcommand{\covmat}{\gamma}
\newcommand{\precision}{\eps}
\newcommand{\overlap}{\eta}

\DeclareMathOperator{\free}{free}
\DeclareMathOperator{\imp}{imp}
\DeclareMathOperator{\poly}{poly}
\DeclareMathOperator{\prob}{P}

\begin{document}

\date{January 14, 2025}

\title{High ground state overlap via quantum embedding methods}

\author{Mihael Erakovic}
\email{mihael.erakovic@phys.chem.ethz.ch}
\affiliation{ETH Z\"{u}rich, Department of Chemistry and Applied Biosciences, Vladimir-Prelog-Weg 2, 8093 Z\"{u}rich, Switzerland}
\author{Freek Witteveen}
\email{fw@math.ku.dk}
\affiliation{Department of Mathematical Sciences, University of Copenhagen, Universitetsparken 5, 2100 Copenhagen, Denmark}
\author{Dylan Harley}
\affiliation{Department of Mathematical Sciences, University of Copenhagen, Universitetsparken 5, 2100 Copenhagen, Denmark}
\author{Jakob G\"unther}
\affiliation{Department of Mathematical Sciences, University of Copenhagen, Universitetsparken 5, 2100 Copenhagen, Denmark}
\author{Moritz Bensberg}
\affiliation{ETH Z\"{u}rich, Department of Chemistry and Applied Biosciences, Vladimir-Prelog-Weg 2, 8093 Z\"{u}rich, Switzerland}
\author{Oinam Romesh Meitei}
\affiliation{Department of Chemistry, Massachusetts Institute of Technology, Cambridge, Massachusetts 02139, United States}
\author{Minsik Cho}
\affiliation{Department of Chemistry, Massachusetts Institute of Technology, Cambridge, Massachusetts 02139, United States}
\author{Troy Van Voorhis}
\affiliation{Department of Chemistry, Massachusetts Institute of Technology, Cambridge, Massachusetts 02139, United States}
\author{Markus Reiher}
\email{mreiher@ethz.ch}
\affiliation{ETH Z\"{u}rich, Department of Chemistry and Applied Biosciences, Vladimir-Prelog-Weg 2, 8093 Z\"{u}rich, Switzerland}
\author{Matthias Christandl}
\email{christandl@math.ku.dk}
\affiliation{Department of Mathematical Sciences, University of Copenhagen, Universitetsparken 5, 2100 Copenhagen, Denmark}

\begin{abstract}
  Quantum computers can accurately compute ground state energies using phase estimation, but this requires a guiding state that has significant overlap with the true ground state.
  For large molecules and extended materials, it becomes difficult to find guiding states with good ground state overlap
  for growing molecule sizes.
  Additionally, the required number of qubits and quantum gates may become prohibitively large.
  One approach for dealing with these challenges is to use a quantum embedding method, which allows a reduction to one or multiple smaller quantum cores embedded in a larger quantum region.
  In such situations it is unclear how the embedding method affects the hardness of constructing good guiding states. In this work, we therefore investigate the preparation of guiding states in the context of quantum embedding methods.
  We extend previous work on quantum impurity problems, a framework in which we can rigorously analyze the embedding of a subset of orbitals.
  While there exist results for optimal active orbital space selection in terms of energy minimization, we rigorously demonstrate how the same principles can be used to define selected orbital spaces for state preparation in terms of the overlap with the ground state. Moreover,
  we perform numerical studies of molecular systems relevant to biochemistry, one field in which quantum embedding methods are required due to the large size of biomacromolecules such as proteins and nucleic acids.
  We investigate two different embedding strategies which can exhibit qualitatively different orbital entanglement.
  In all cases we demonstrate that the easy-to-obtain mean-field state will have a sufficiently high overlap with the target state to perform quantum phase estimation.
\end{abstract}

\maketitle

\section{Introduction}

Some of the most promising applications for quantum computers arguably lie in their utility for molecular and materials science.
These fields face important computational challenges, in both academic and industrial settings.
Prominent examples can be found in catalysis, battery and drug design \cite{reiher2017elucidating,von2021quantum,kim2022fault,santagati2024drug}, and biochemistry \cite{goings2022reliably,baiardi2023quantum}.
True quantum advantage will likely only emerge in the fault-tolerant regime, where quantum phase estimation (QPE) algorithms will allow for energy calculations of a quantum system (i.e., a molecule or material) with controlled and guaranteed accuracy.

One of the central problems of quantum chemistry is the computation of ground-state electronic energies.
We consider a system with a fixed number $n$ of electrons, and a discretization of the electronic structure problem in the Born-Oppenheimer approximation into $N$ spin-orbitals, giving a Hamiltonian of the form
\begin{align*}
  H = \sum_{pq = 1}^N  h_{pq} a_{p}^{\dag} a_{q}
  + \frac{1}{2}\sum_{pqrs=1}^N g_{pqrs} a_{p}^{\dag} a_{r}^{\dag} a_{s} a_{q} \, ,
\end{align*}
where $a_p$ for $p = 1,\dots,N$ are the fermionic annihilation operators for the $N$ spin-orbitals and $g_{pqrs} = (pq | rs)$ are the two-body integrals, see for example Ref.~\cite{helgakerMolecularElectronicStructureTheory2014} for details.
A system with $N$ spin-orbitals and $n$ electrons has a Hilbert space of dimension ${N \choose n}$ which is exponential in $N$ if $n$ scales linearly with $N$.
High accuracy classical methods struggle with the size of this Hilbert space.
The key advantage of quantum computing is that the quantum computer can natively represent states in the Hilbert space using only $N$ qubits.
When considering the scaling with system size, one should keep in mind that at fixed electron number $n$, but an increasing number of orbitals $N$ (the continuum limit), the Hilbert space dimension increases only as $\poly(N)$ (polynomially in $N$), albeit with an in practice prohibitive exponent $n$.

There are two major challenges in quantum computing for chemistry and many-body physics.
\begin{enumerate}
  \item The first challenge is the \emph{orthogonality catastrophe}. Quantum phase estimation requires an initial \emph{guiding state} which has sufficient
        overlap with the ground state. However, for systems with large $N$ and scaling electron number, small local errors in the guiding state lead to an exponential decay in the overlap with the global ground state. As a result good guiding states become hard to find, see \cref{sec:orthogonality}.
  \item The second challenge is to perform Hamiltonian simulation for a time $\bigO(\eps^{-1})$ to achieve quantum phase estimation precision $\eps$, and to reduce the dependence of the computational cost on $N$.
\end{enumerate}

The first challenge is a fundamental and \emph{theoretical} obstruction to computing ground state energies efficiently on a quantum computer.
However, for many molecules, conventional computational methods are observed to find good guiding states in practice (though this may be different for strongly correlated materials with a lattice structure), so for practical applications to quantum chemistry there is evidence that this is not the main problem \cite{tubmanPostponingOrthogonalityCatastrophe2018,fomichevInitialStatePreparation2023} and our work supports this claim.
For the second challenge, phase estimation based on Trotterization of the time evolution unitary operator scales as $\mathcal{O}(N^4)$ for the full electronic Hamiltonian because of the two-electron interaction terms in its second-quantized form. However, truncation strategies can be employed to reduce the number of relevant parameters and significantly reduce the scaling of the phase estimation. In the case of Trotterization, the scaling can be reduced to $\mathcal{O}(N^3)$ for a given molecule and increasing basis set size and to $\mathcal{O}(N^2\log N)$ for increasing molecule size \cite{motta2021low_rank}. Alternatively, qubitization allows for the factorization of the Hamiltonian, which can reduce the scaling down to $\mathcal{O}(N)$ \cite{von2021quantum}. However, this comes at the cost of a large number of ancillary qubits, and the overall scaling of the phase estimation also depends on a normalization factor $\lambda$ which scales with $N$. Hence, the scaling can lead to quantum circuits with a prohibitively large gate count
in practice, even for different truncation and factorization strategies, and this cost represents a serious bottleneck for useful quantum computations. Specifically, from state-of-the-art resource estimates it appears that, with fault-tolerant devices, computations for 100 to 200 spin orbitals already require very extensive quantum resources \cite{von2021quantum,lee2021even,kim2022fault,blunt2024resource_estimates}, on the order of $10^{10}$ Toffoli gates \cite{berry2019qubitization} (with the caveat that it is difficult to predict the hardware specifications of future fault-tolerant quantum computers).

Given a choice of orbital basis, the electronic ground state can be written as
\begin{align}\label{eq:basis expansion}
  \ket{\Psi} = \sum_{i} C_i \ket{\Phi_i}, \quad \sum_i \abs{C_i}^2 = 1,
\end{align}
where the sum runs over Slater determinants $\ket{\Phi_i}$.
For example, one can choose Hartree-Fock (HF) orbitals.
In that case, the expansion in \cref{eq:basis expansion} consists of the HF state $\ket{\Phi_0}$ and corrections to it, which represent electron correlations.
One can also choose a different orbital basis, and this may lead to a significant increase in the largest weight on a single determinant.
If the HF basis state qualitatively misrepresents the ground state, and there are more determinants with weights of the same order of magnitude, this is known as \emph{static correlation}.
The remaining determinants with small weights then account for the so-called \emph{dynamical correlation}.
While useful notions, there is no sharp distinction between dynamical and static correlation.

The question of finding and preparing good guiding states is closely related to the computational complexity of the electronic structure problem.
The problem of finding ground state energies of local Hamiltonians in full generality is QMA-hard \cite{kempe2006complexity,ogorman2022intractability}, so this task is believed to be intractable even for quantum computers. In particular, this implies that it is hard to prepare guiding states with non-trivial ground state overlap --- however, numerical evidence suggests that finding good guiding states is feasible for many realistic not-too-large chemical systems \cite{tubmanPostponingOrthogonalityCatastrophe2018,leeEvaluatingEvidenceExponential2023}.
This is the case if the correlation is of a dynamical nature, and also in systems with a static correlation; a small number of determinants \cite{fomichevInitialStatePreparation2023}, choosing a different orbital basis (see, for instance, Ref.~\cite{ollitrault2024enhancing}), or symmetry-respecting configuration state functions \cite{marti2024spin} may suffice to reach good ground state overlap and hence yield efficient quantum phase estimation algorithms.
For quantum phase estimation, a constant (or even inverse polynomially small) ground state overlap suffices for efficient ground state energy computation.
For the total cost, it does not make a significant difference whether one has a constant ground state overlap, or even ground state overlap close to 1.
However, if the guiding state has ground state overlap close to 1, this allows the possibility to paralellize the phase estimation into multiple circuits of shorter depth \cite{ding2023even,ni2023low}, as we briefly review in \cref{sec:orthogonality}, and which may be particularly useful for early fault-tolerant quantum computers.

In many large molecules where static quantum correlations are important, so that a mean-field treatment does not suffice to represent the target state even qualitatively well, these correlations can be assigned to a relatively small number of orbitals. If these orbitals are localized in three-dimensional space, static correlation will be of a short-ranged nature.
In such situations, both challenges (ground state overlap and Hamiltonian simulation scaling with $N$) can be addressed by using an appropriate \emph{embedding method}, which singles out a spatial region (or length scale) that is treated fully quantum mechanically, and an environment (or longer range correlations) which can be treated with an approximate electronic structure model such as (mean-field) HF or some other electronic structure model of low computational complexity such as (variants of) Kohn-Sham density functional theory (DFT).
By focusing on a smaller embedded quantum subsystem (denoted as a quantum core in the following), the system to which one applies phase estimation is sufficiently small so that the exponential scaling of the orthogonality catastrophe has not yet kicked in, and one can still find high-overlap guiding states using classical methods.
Additionally, the smaller number of orbitals leads to a reduced gate count in the phase estimation quantum circuit.
For very large molecules (such as proteins, protein complexes, or interacting biomolecules in biochemistry), embedding strategies can be of a multi-level nature, where the description of the interatomic interactions ranges from classical force fields to different quantum mechanical approximations.
We provide a brief introduction to relevant aspects of quantum embedding theory in \cref{sec:embedding}.

In this work we focus on quantum-in-quantum embedding schemes, where a large system is described by an approximate quantum mechanical model (such as HF or Kohn-Sham DFT),
while one or more subsystems embedded into this large system define the quantum cores. The energy contribution of the quantum cores to the full potential energy surface of the large system can be rigorously obtained by a (future) fault-tolerant quantum computer using a representation of the full wavefunction in the restricted orbital space of the quantum cores.
It is for this reason that we consider state preparation in the context of quantum embedding approaches.

While wavefunction-based approaches in traditional computing (such as coupled cluster theory or multi-configurational approaches with multi-reference perturbation theory) can deliver accurate quantum core energies, an important drawback is that precise and controllable error bounds are not known for any of these methods \cite{reiher2022molecule}. By contrast, quantum computation based on phase estimation can deliver total electronic energies to a given precision (that is, typically chemical accuracy between 1 and 0.1 mHartree, %\sout{) --- provided that an initial guiding state can be efficiently prepared which will have a large overlap with the target state that is to be determined.}
which is sufficient to ensure that relative energies are sufficiently accurate for the evaluation of valence-shell properties such as relative  energies of molecular structures or rate constants) --- provided that an initial guiding state can be efficiently prepared which will have a large overlap with the target state that is to be determined. We note that the growth of the absolute electronic energy with molecular size is due to the low-lying core shell orbitals, which do not contribute to such valence shell properties (apart from the fact that an embedding that restricts the orbital space would not allow for arbitrary growth of the electronic energy).
We emphasize that rigorous error estimates will be a key advantage over traditional approaches \cite{reiher2017elucidating, Liu2022,von2021quantum},
apart from the fact that a quantum computer with a sufficiently large number of qubits for the representation of an electronic state will tame the curse of dimensionality posed to traditional approaches.

The use of embedding methods for quantum computing in quantum chemistry raises important questions:
\begin{enumerate}
  \item What is the interplay between an embedding method and the \emph{guiding state}? For example, if the choice of the embedding method affects the type of correlation on the resulting orbital space in a quantum core, then this may have consequences for the character of the guiding state up to the point where it might be difficult to determine.
  \item What is the \emph{computational complexity} of problems with localized quantum correlations?
\end{enumerate}

The second question concerns the general \emph{computational complexity} of problems with localized quantum correlations.
In \cref{sec:impurity}, we study this question theoretically in the context of \emph{quantum impurity models}, specifically by extending the work of \cite{bravyi2017complexity} on the nature of \emph{quantum impurity problems} to shed light on these questions.
A quantum impurity model is a system with a scaling number of $N$ orbitals, of which only a constant number $M$ participate in two-body interactions.
Such models may be considered a useful description for systems where electron correlation is localized to a small subset of orbitals.
Based on \cite{bravyi2017complexity}, we show that there exist good quantum embeddings for such systems.
Based on this, we propose an approach to the guiding state problem where one uses an embedding method to find a guiding state that is an arbitrary state on the quantum core, and a Slater determinant on the environment.
We show that if one chooses a good embedding or if the one-body Hamiltonian is gapped, this guiding state has at least an inverse polynomial overlap with the true ground state of the embedding Hamiltonian, giving a polynomial quantum algorithm for the ground state energy.
This is in contrast to the best-known rigorous classical algorithms, which in these cases have a quasi-polynomial scaling in the precision \cite{bravyi2017complexity}.

These results motivate a general strategy for preparing guiding states for problems with localized quantum correlations on large-scale fault-tolerant quantum computers.
We propose that one searches for guiding states through quantum-in-quantum embeddings.

Here, one identifies an active orbital space $A$ and an environment $C$, and takes as a guiding state a state of the form $\ket{\Phi_A} \wedge \ket{\Psi_C}$, where $\ket{\Phi_A}$ is the solution to the (small) active space problem, $\ket{\Psi_C}$ is a Slater determinant on the environment, and $\wedge$ denotes the antisymmetrized product. A mean-field calculation is used to obtain the wavefunction on the environment, with which the embedded Hamiltonian can be constructed. An approximation of the ground state energy and wavefunction can then be calculated using one of the classical approaches. Tensor network-based methods, such as the density matrix renormalization group (DMRG), stand out in particular as they provide an adequate description even in the case of strongly correlated embedded systems. These approaches, however, introduce errors in the energy due to approximations (e.g. limited bond dimension in the case of DMRG), but may provide a good approximation to the wavefunction. Quantum phase estimation then allows for accurate computation of the ground state energy using the clasically calculated approximate wavefunction as a guiding state.

We demonstrate numerically in \cref{sec:numerical} that for two conceptually different embedding methods, one can find low-complexity guiding states, even though the two different methods produce rather different orbital bases. For this demonstration we focus on small biomolecules (tryptophan and its oligopeptide structures), which occur as monomers and oligomers in biomacromolecules (proteins), and on a transition-metal drug that can form a host-guest complex with a protein. Such complexes are key to molecular recognition processes.
While we focus on biomolecules, our approach is similarly relevant in other application areas involving larger molecules or materials.

\section{Quantum embedding methods}\label{sec:embedding}
In order to make computations feasible for large systems, for which physical or chemical properties are governed by relatively small spatial regions within the larger system, an established approach (see, for instance, Refs.~\cite{warshel1976theoretical,wesolowski1993frozen})
is to treat different subsystems or length scales at different levels of accuracy.
For example, in the case of molecular recognition, it may be sufficient to focus on the interaction seam of host and guest molecules and, hence, to take into account electron correlation accurately only in the binding region.
\emph{Embedding theory} deals with approaches to embed
approximations of different levels of accuracy into one another;
see Refs.~\cite{Senn2009,Mennucci2012,Libisch2014,sun2016quantum,Jones2020,Jacob2024} for reviews of quantum embedding theory.
In general, one decomposes the system into a \emph{quantum core}, where one uses a \emph{high-level} electronic structure model, which takes into account the more complex nature of the wavefunction, and a (quantum or classical) environment. If feasible, the environment can be treated using a \emph{low-level} model such as HF or DFT.
It is possible to have multiple quantum cores and multiple embedding layers, treating parts of the system at the level of classical force fields, semi-empirical methods, or DFT at different levels of accuracy.
In order to take the effect of the environment into account, an embedding scheme requires the construction of an effective Hamiltonian on the quantum cores, given the solution of the low-level solver.
Given the result of minimizing this Hamiltonian on the quantum cores, the embedding scheme should provide a prescription for a global energy estimate.

We will discuss two different embedding schemes, which we compare in detail with molecular examples in \cref{sec:numerical}.
The first approach is based on density matrices.
This is natural for quantum computing, since quantum computation naturally interfaces with quantum states.
For example, in the bootstrap embedding method (a particular density matrix-based embedding scheme), one can enforce the consistency conditions between overlapping fragments on the quantum computer \cite{liu2023bootstrap}.
As an alternative, we also consider projection-based embedding methods.
These have the advantage of being naturally compatible with DFT, which may be useful for multiscale approaches.

\subsection{Embedding based on density matrices}
A natural approach toward embedding emerges from the framework of open-system quantum theory and is based on density matrices; an example is \emph{density matrix embedding theory} (DMET) \cite{knizia2012density,wouters2016practical,pham2018can, Sekaran2021},
see \cite{cances2023some} for a precise mathematical description.
Here we consider a system with $N$ orbitals, where we identify subsets $\act_1, \dots, \act_n$ of these orbitals, which may be called `fragments' or `impurities'. We let $M_i$ be the number of orbitals in $\act_i$. 
Depending on the method, these fragments may or may not overlap.
Whereas the formal basis of DMET is the Schmidt decomposition, its
key practical idea is to combine a high-level solver on the fragments with a mean-field solution on the global level and enforce consistency conditions.
We start with a global state given by a single Slater determinant $\ket{\Psi}$, which can be obtained as the Hartree-Fock solution by the low-level solver.
For each of the subsystems $\act_i$ we can perform a Schmidt decomposition, which for each $i$ gives a decomposition of the full $N$-dimensional single-particle space into $\act_i$ (the fragment orbitals), a `bath' $\bath_i$ for $\act_i$ of size at most $M_i$ (the bath orbitals) and its complement $\env_i$, such that the state
\begin{align*}
  \ket{\Psi} = \ket{\Psi_{\act_i \bath_i}} \wedge \ket{\Psi_{\env_i}}
\end{align*}
is decomposed into Slater determinants $\ket{\Psi_{\act_i \bath_i}}$ and $\ket{\Psi_{\env_i}}$ on $\act_i \bath_i$ and $\env_i$ respectively.
Now, one may look for wavefunctions of the form $\ket{\Psi_i} = \ket{\Phi_{\act_i \bath_i}} \wedge \ket{\Psi_{\env_i}}$ where $\ket{\Phi_{\act_i \bath_i}}$ is now an arbitrary fermionic state on $\act_i \bath_i$, determined by a high-precision solver (such as a quantum computer).
The high-level solver minimizes the energy over states of this form.
Since we have a fixed Slater determinant $\ket{\Psi_{\env_i}}$ on $\env_i$, this gives an effective Hamiltonian on $\act_i \bath_i$
\begin{align*}
  H_{\act_i \bath_i} = \left(I \ot \bra{\Psi_{\env_i}}\right) H \left(I \ot \ket{\Psi_{\env_i}}\right).
\end{align*}
This Hamiltonian is such that
\begin{align*}
  \min_{\ket{\Psi_i}} \bra{\Psi_i} H \ket{\Psi_i} & = \min_{\ket{\Phi_{\act_i \bath_i}}} \bra{\Phi_{\act_i \bath_i}} H_{\act_i \bath_i}\ket{\Phi_{\act_i \bath_i}},
\end{align*}
where the minimization on the left-hand side is over states of the form
\begin{align*}
  \ket{\Psi_i} = \ket{\Phi_{\act_i \bath_i}} \wedge \ket{\Psi_{\env_i}}.
\end{align*}

There are different schemes for estimating a global energy based on density matrix embedding.
The standard approach is to try to achieve self-consistency between the global Slater determinant and the fragment wavefunctions.
Here, one takes non-overlapping fragments and searches for a solution where the Slater determinant $\ket{\Psi}$ has the same 1-RDMs on the fragments $\act_i$ as the fragment solutions $\ket{\Phi_{\act_i \bath_i}}$.
This can be obtained by iteratively solving for the fragments and updating them to a global mean-field solution with constrained 1-RDMs.
A different approach is the \emph{bootstrap embedding}, where one takes overlapping fragments and iteratively tries to enforce consistency between the fragments on their overlap instead of with a global state, see \cite{welborn2016bootstrap,ye2019bootstrap,ye2020bootstrap} for details.
However, we note that in both cases, there is no guarantee of the existence of a \emph{global} wavefunction compatible with these schemes, and the results are not variational. This means that the obtained value is not necessarily a strict upper bound on the true ground state energy.

\subsection{Embedding based on projection}
An alternative approach is given by \emph{projection-based} embeddings \cite{hegely2016exact, Manby2012}.
Here, we use a method based on the Huzinaga equation \cite{huzinaga1971theory,hegely2016exact}.
We describe the case with a single quantum core.
We start the theory discussion with HF as the low-level solver, highlighting the similarities to DMET. Later will we describe how DFT retrieves the missing electron correlation in Hartree-Fock, making the approach formally exact.
The exactness of the approach means that when using the exact (unknown) exchange correlation functional for the DFT computation, and the exact wavefunction solution on the embedded fragment, the result of following the embedding procedure equals the exact ground state energy.
This embedding method starts from a HF (or from a Kohn-Sham density functional theory calculation), yielding a collection of molecular orbitals indexed by a set $J$ and a single Slater determinant $\ket{\Phi}$ with the orbitals in $\mathcal N \subseteq J$ occupied.
Given these orbitals $\phi_j$, they are partitioned into two sets $\act$ and $\env$ of active and environment orbitals.
We write $\mathcal N_{\act}$ and $\mathcal N_{\env}$ for the subsets of occupied orbitals in $\ket{\Phi}$.
The division into $\act$ and $\env$ is often based on the localization of the orbitals on specific atoms in the molecule \cite{Manby2012, hegely2016exact}. Alternatively, a unitary transformation in $\mathcal N$ may be transformed to enforce a localization on specific atoms \cite{Claudino2019}, or an automated orbital analysis can be performed to identify orbitals contributing strongly to energy differences \cite{Bensberg2019a}.
The idea of the projection-based embedding is now again to try to solve the minimization problem
\begin{align*}
  \min_{\ket{\Psi}} \bra{\Psi} H \ket{\Psi},   \qquad \ket{\Psi} = \ket{\Phi_{\act}} \wedge \ket{\Psi_{\env}},
\end{align*}
where $\ket{\Phi_{\act}}$ is an arbitrary state on $\act$ and $\ket{\Psi_{\env}}$ is the Slater determinant on the environment where the orbitals in $\mathcal N_{\env}$ are occupied. Similarly to the embedding based on density matrices, the Slater determinant on the environment is kept fixed, which yields the effective Hamiltonian for the fragment
\begin{align*}
  H_{\act} & = \left(I \ot \bra{\Psi_{\env}}\right) H \left(I \ot \ket{\Psi_{\env}}\right)                                   \\
           & = \bar{E} + \sum_{p,q \in \act} h^{(\act \text{in} \env)}_{pq}a^\dagger_p a_q + \frac{1}{2} \sum_{\substack{p,q \\ r,s \in \act}} (pq|rs) a^\dagger_p a^\dagger_r a_s a_q,
\end{align*}
where
\begin{align*}
  \bar{E}=2\sum_{p\in\mathcal N_{\env}}h_{pp} + \sum_{p,q\in\mathcal N_{\env}}[ 2(pp|qq)-(pq|qp)]
\end{align*}
is the mean-field energy of the environment, and
\begin{align*}
  h^{(\act \text{in} \env)}_{pq}=h_{pq} + \sum_{c\in\mathcal N_{\env}}[ 2(pq|cc)-(pc|qc)]
\end{align*}
is the effective one-electron operator that includes the interaction between the fragment and the environment. In this case, both the environment energy and the fragment-environment interactions are captured on a mean-field level only, neglecting the correlation, while the overall energy is represented as $E=\bar{E}+E^{\act\text{in}\env}$, where $E^{\act\text{in}\env}$ is the energy of $h^{(\act \text{in} \env)}$.
Note that the Hamiltonian $H_A$ on the quantum core has as two-body terms precisely the two-body terms acting on $A$, and only the one-body term is modified.

The key idea of the projection-based embedding is to recover the missing correlation using DFT. For this purpose, densities $\covmat$ from the initial DFT calculation are used to define the interaction terms between the fragment and the environment as
\begin{align*}
  v_{\text{emb}}[\covmat^\act,\covmat^\env] & = g[\covmat^\act+\covmat^\env] - g[\covmat^\act],
\end{align*}
where $g$ is defined by
\begin{align*}
   & (g[\covmat])_{pq}                         = \sum_{ij} \covmat_{ij}\left[(pq|ij)-\frac{1}{2}x(pi|qj)\right] + \left(v_{\text{xc}}[\covmat]\right)_{pq}.
\end{align*}
Here, $x$ is the fraction of the exact exchange in the exchange-correlation functional and $v_{\text{xc}}$ is the exchange-correlation potential matrix. This interaction term is used to redefine the effective one-electron operator as
\begin{align*}
  \tilde{h}^{(\act \text{in} \env)}_{pq} = & h_{pq} + \left(v_{\text{emb}}[\covmat^{\act},\covmat^\env]\right)_{pq}  \\
                                           & - \left((F^\act - \mu) P^{\env} - P^{\env} (F^\act - \mu)\right)_{pq} ,
\end{align*}
where $P^{\env}$ is the matrix of the projection operator onto the subspace $\env$, $\mu$ is a positive constant energy shift and $F^\act$ is the Fock matrix for the embedded fragment. The last two terms are used to ensure that there is no mixing of the fragment and the environment orbitals. The overall energy of the system can be written as
\begin{align*}
  E = & \bra{\Psi_\act}\tilde{H}_\act \ket{\Psi_\act} -\bra{\Psi_\act} \hat{v}_{\text{xc}}[\covmat^\act+\covmat^\env] - \hat{v}_{\text{xc}}[\covmat^\env]\ket{\Psi_\act} \\
      & + E_{\text{xc}}^{\text{DFT}}[\covmat^\act+\covmat^\env] - E_{\text{xc}}^{\text{DFT}}[\covmat^\act] - E_{\text{xc}}^{\text{DFT}}[\covmat^\env]                    \\
      & + E^{\text{DFT}}[\covmat^\env],
\end{align*}
where $E^{\text{DFT}}[\covmat^\env]$ is the Kohn-Sham DFT energy of $\env$, $E_{\text{xc}}^{\text{DFT}}[\covmat]$ is the exchange-correlation functional evaluated for the density $\covmat$, and $\hat{v}_{\text{xc}}[\covmat]$ is the operator belonging to $v_{\text{xc}}$.
Note that $E_{\text{xc}}^{\text{DFT}}[\covmat]$ is typically nonlinear in $\covmat$, which means that it cannot be calculated as the expectation value of $\hat{v}_{\text{xc}}[\covmat]$. Therefore, we formally subtracted the term $\bra{\Psi_\act} \hat{v}_{\text{xc}}[\covmat^\act+\covmat^\env] - \hat{v}_{\text{xc}}[\covmat^\env]\ket{\Psi_\act}$ and added the differences in $E_{\text{xc}}^{\text{DFT}}[\covmat]$.

\subsection{Embedding methods and quantum computing}
Using appropriate embedding methods will be crucial for the application of quantum computers to chemical systems, especially for biochemistry, nanochemistry, and materials science.
Identifying the right quantum mechanical subsystem of a molecule and using the quantum chemistry calculation on this fragment to derive conclusions about the full system is unavoidable for practical applications of quantum computation.
The quantum computer then serves as the high-level solver in the embedding scheme.

Previous work using density matrix-based embedding schemes includes Refs.~\cite{mineh2022solving,tilly2021reduced,kawashima2021optimizing,liu2023bootstrap,cao2023ab}, approaches using projection-based embeddings include Refs.~\cite{ralli2024scalable, rossmannek2023quantum, battaglia2024general}, and approaches based on Green's functions are presented in Refs.~\cite{vorwerk2022quantum,bauer2016hybrid} (which we do not discuss in this work).
These previous works have mostly focused on reducing the problem size, and running heuristic variational algorithms on (near-term) quantum computers.
How the embedding method influences the problem of finding good guiding states, and thereby the prospects for phase estimation, has not been explored systematically in previous work to the best of our knowledge; one work in this direction is \cite{goings2022reliably} which computes guiding state overlaps for a range of active space sizes for phase estimation for a protein system.

\section{Orbital selection for state preparation in quantum impurity problems}\label{sec:impurity}

In this section, we aim at deriving rigorous results for the state preparation problem in an embedding situation. To to accomplish this, we exploit the framework of quantum impurity problems (see \cref{fig:impurity}).
A \emph{quantum impurity model} is a fermionic Hamiltonian which has arbitrary one-body interactions, and two-body interactions on a constant-sized subset of orbitals.
Quantum impurity models are relevant as models in material science where there is an impurity in the material.
A famous example is the Anderson impurity model \cite{anderson1961localized} which describes a magnetic impurity in a metal. This model gives an explanation for the Kondo effect and has been instrumental in the development of numerical renormalization group techniques \cite{wilson1975renormalization}.
In addition to the important application to impurities in metals, one can also consider them as a toy model for large molecules where mean-field approximations are accurate outside of a small interaction region.
For example, the `impurity' may consist of the orbitals that are localized in the binding region.
Note that in the quantum impurity model, we assume that away from the impurity we have a one-body Hamiltonian (which is a different assumption from having an accurate mean-field approximation).

More formally, a quantum impurity model is defined by a fermionic Hamiltonian of the form $H_{\free} + H_{\imp}$ where $H_{\free}$ is a one-body Hamiltonian on $N$ fermionic modes, whereas $H_{\imp}$ is an interacting fermionic Hamiltonian acting on a subset of $M$ modes
\begin{align*}
  H_{\free} & = \sum_{pq = 1}^N h_{pq} a_p^\dagger a_q,                     \\
  H_{\imp}  & = \sum_{pqrs = 1}^M h_{pqrs} a_p^\dagger a_q^\dagger a_r a_s.
\end{align*}
Typically we consider constant $M$ and a scaling total system size $N$.
We assume that $H_{\free}$ has single-particle energies in a constant range (which we can always achieve by rescaling).
We write $\gap$ for the ground state energy gap of $H_{\free}$.
If the Hamiltonian is number-conserving, we can study the lowest-energy state in the subspace where we fixed the number of particles to $n$.

\begin{figure}
  \centering
  \includegraphics[scale=0.5]{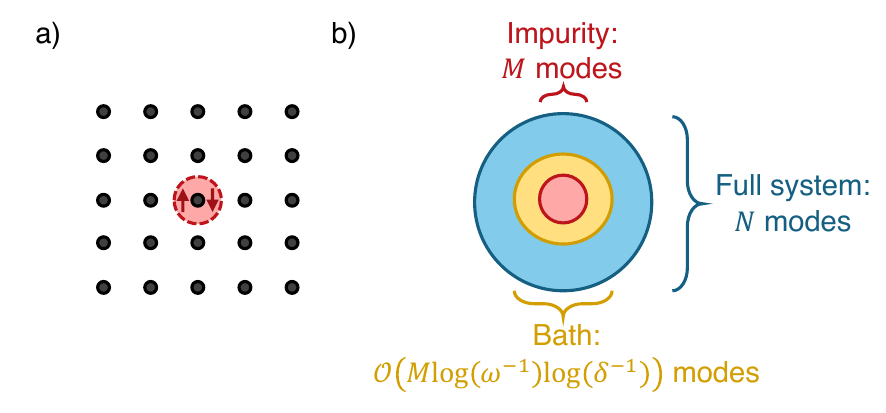}
  \caption{(a) A lattice model for the electrons in a metal, with an impurity at one site.
    (b) Illustration of \cref{thm:Slater det impurity}: the ground state can be accurately approximated by a state which is a Slater determinant on the environment and some unconstrained state on the fragment and bath modes.}
  \label{fig:impurity}
\end{figure}

Quantum impurity models are a simple but nontrivial class of models that capture the orthogonality catastrophe and the quality of embedding methods.
The Anderson derivation of the orthogonality catastrophe \cite{anderson1967infrared}, see \cref{sec:orthogonality}, suggests that the scaling of the orthogonality catastrophe may be polynomial rather than exponential for quantum impurity problems. For quantum phase estimation an inverse polynomial overlap suffices to get a polynomial algorithm, so this suggests that quantum impurity problems may be in BQP.

Secondly, from the perspective of embedding methods, since the impurity is small compared to the system size, one may hope that one can identify a \emph{bath} of a relatively small number of orbitals that are influenced by the impurity, and one can solve the full system by restricting to this bath.

Quantum impurity problems have been studied rigorously by Bravyi and Gosset in \cite{bravyi2017complexity}. Their main technical result is a bound on the decay of the eigenvalues of the 1-RDM.
This gives rise to approximations of the ground state by a relatively small number of Gaussian states.
As an application, Ref.~\cite{bravyi2017complexity} gives a quasipolynomial classical algorithm for the ground state energy, and shows that the quantum impurity problem is in QCMA, which is the class of problems that can be efficiently verified by a quantum computer given a classical witness, and which is between NP and QMA.

Our results build on this work in two ways. Firstly, we argue that under certain conditions one can find an efficiently preparable state that has inverse polynomial overlap with the ground state, giving an efficient quantum algorithm for the ground state energy.

Secondly, the results in \cite{bravyi2017complexity} are formulated in terms of Gaussian states rather than Slater determinants.
In order to relate their results to usual approaches to embedding methods in quantum chemistry, we derive the analagous statements using Slater determinants, working at a fixed particle number.
A detailed description of the set-up, as well as formal statements and proofs are provided in \cref{sec:proof impurity}.

\subsection{Decay properties of the one-body reduced density matrix and good embeddings}
One-electron basis states (orbitals) are key to understanding different embedding techniques \cite{Muehlbach2018}.
Given an electronic Hamiltonian on a set of orbitals and a fixed number of electrons, orbitals are often partitioned into sets that are called \emph{frozen} (doubly occupied), \emph{active}, and \emph{virtual} (unoccupied).
The idea is that, in an expansion of the wavefunction as in \cref{eq:basis expansion},
only the frozen and a few additional orbitals (that is, the active ones) give rise to determinants with large
weights $\alpha_i$.
This means that the problem reduces to an effective problem defined only on the active space.
The frozen and virtual spaces may contribute to the dynamical correlation, but this can be accounted for by a computationally less intensive method (for example, by perturbation theory).
The selection of the active orbitals can be based on various measures such as natural orbital occupation numbers \cite{jensen1988MP2NO,bofill1989UNO,krausbeck2014NOexcited,keller2015UNO} or information entropy based \cite{legeza2003entropy,rissler2006entropy,boguslawski2012entropy,stein2016automated}.
Our formal analysis will focus on the former in the following, whereas for practical reasons, our numerical results will exploit the latter, which can be based on localized orbitals that can have advantages over natural orbitals. In any case, as we review in \cref{sec:orthogonality}, orbital rotations are always possible, which would allow one to exploit natural orbitals for state preparation, which can then be rotated into a local basis.

Given a ground state $\ket{\Psi}$, the \emph{one-body reduced density matrix} (1-RDM), or covariance matrix, is the matrix $\covmat$ defined by $\covmat_{pq} = \bra{\Psi} a_p^\dagger a_q \ket{\Psi}$.
Its eigenvectors define a set of orbitals known as the \emph{natural orbitals}.
We order the eigenvalues of $\covmat$ as $1 \geq \sigma_1 \geq \sigma_2 \geq \dots \geq \sigma_N \geq 0$. The eigenvalue $\sigma_j$ equals the expected particle number for natural orbital $j$.
This means that the natural orbitals for which $\sigma_j$ equals $0$ or $1$ are always unoccupied or occupied respectively in the ground state. If the eigenvalue $\sigma_j$ is close to $0$ or $1$, we may approximate the ground state by a state where the corresponding natural orbital is unoccupied or occupied respectively.
In other words, if we can bound the number of orbitals for which $\sigma_j$ is \emph{not} close to $0$ or $1$ by $K$, then we can approximate the ground state $\ket{\Psi}$ by a state $\ket{\tilde\Psi} = \ket{\Phi} \wedge \ket{\Theta}$, where $\ket{\Theta}$ is a Slater determinant and $\ket{\Phi}$ is an arbitrary state on $K$ modes.
This means that these $K$ modes define a good active space for the problem at hand, and we can reduce the problem of (approximately) computing the ground state energy to this subspace.

The main result of \cite{bravyi2017complexity} is a bound on the eigenvalues of the 1-RDM of quantum impurity problems, which is used to approximate the ground state in terms of Gaussian states.
For chemistry problems we typically have a fixed number of electrons, and it is more convenient to work with Slater determinants.
We adapt the argument of \cite{bravyi2017complexity} to show the following result, which we prove in \cref{sec:proof impurity}.

\begin{thm}\label{thm:Slater det impurity}
  Let $\gap > 0$ be the ground state energy gap of $H_{\free}$, and let $\eps > 0$. Then for
  \begin{align*}
    K = \bigO\mleft(\log(\gap^{-1}) \mleft(\log(\eps^{-1}) + \log \log(\gap^{-1}) \mright) \mright)
  \end{align*}
  there exists a Slater determinant $\ket{\Theta}$ on $N - K$ modes and an arbitrary state $\ket{\Phi}$ on $K$ modes such that the state
  \begin{align*}
    \ket{\tilde\Psi} = \ket{\Phi} \wedge \ket{\Theta}
  \end{align*}
  has overlap $\abs{\braket{\tilde\Psi}{\Psi}} \geq 1 - \eps$ with a ground state $\ket{\Psi}$ of $H_{\free} + H_{\imp}$.
\end{thm}

Note that while this result shows that there \emph{exists} a good active space, it only provides a good way to determine which orbitals make up this active space given the ground state 1-RDM (which one does not necessarily have access to without a priori knowledge of the ground state).
In practice, finding good active spaces is based on using an efficient classical method to approximate the 1-RDM and using this estimate to select an appropriate accurate space as discussed in \cref{sec:embedding}.

\subsection{Computational complexity of the quantum impurity problem}
The general electronic structure problem is QMA-complete \cite{schuch2009computational,ogorman2022intractability}, meaning that it is likely a hard problem even for quantum computers.
On the other hand, Hamiltonians with only one-body interactions, like $H_{\free}$, can be efficiently simulated classically \cite{bravyi2004lagrangian}.
An important broad question in quantum computing is to determine the existence of (physically relevant) families of Hamiltonians for which the ground state problem is in BQP, while computing the ground state energy is hard classically.
Proving such a separation is difficult (as this would imply a separation between P and BQP), but one can at least try to prove BQP containment for specific families of Hamiltonian ground state energy problems for which we do not know rigorous polynomial-time classical algorithms, see for example \cite{chen2024sparse,hastings2022optimizing}.

We investigate the quantum computational complexity of quantum impurity problems.
Here, the goal is to approximate the ground state energy of an impurity problem on $N$ modes, with a constant-sized impurity, to precision $\eps$.
As we saw in the previous section, for quantum impurity problems there exist relatively small active spaces, significantly reducing the size of the problem.
It was shown in \cite{bravyi2017complexity} that this can be leveraged to give a \emph{classical} algorithm that scales polynomially in $N$ and quasipolynomially in the precision $\eps$, with overall scaling $\bigO(N^3) \exp(\bigO(\log(\eps^{-1})^3))$.
It is also shown that the ground state energy problem for quantum impurity models is contained in QCMA, the class of problems for which a quantum computer can efficiently verify solutions given a classical witness. QCMA-hard problems are unlikely to be solvable efficiently on a quantum computer.
This leaves the interesting open problem whether one can prove whether quantum computers can efficiently approximate the ground state energy of quantum impurity problems to polynomial precision.
We argue that under two different restrictions, quantum impurity problems become easy for quantum computers (while not obviously so for classical algorithms).

The main observation is that when we apply \cref{thm:Slater det impurity}, it suffices to find an active space such that the ground state overlap is constant (but not close to 1).
If this active space contains only a constant number of orbitals $K$, then we may consider the state $\rho$ corresponding to choosing a random active space of size $K$ and preparing the maximally mixed state for this active space.
For constant $K$, this leads to a ground state overlap that is at least $\poly(N)^{-1}$, and hence $\rho$ can be used as a guiding state to yield a polynomial time algorithm.
Indeed, \cref{thm:Slater det impurity} gives a constant-sized active space for constant precision if the gap $\gap$ is constant.
As a further comment, these assumptions change the scaling of the classical algorithm of \cite{bravyi2017complexity} to $\bigO(N^3) \exp(\bigO(\log(\eps^{-1})^2))$ in the case of constant gap.
This provides evidence that quantum impurity models may provide an example class of systems with polynomial time quantum algorithm, but no polynomial time classical algorithm. However, the precise status of this remains unclear, in particular since there is already a classical algorithm with quasipolynomial scaling, and an improved analysis or better classical algorithm could lead to polynomial scaling.
Moreover, the potential advantage disappears if the full quantum impurity Hamiltonian is gapped, as \cite{bravyi2017complexity} provides a polynomial classical algorithm for that situation based on matrix product states (MPS).

A second simplification arises in the case where we have a sufficiently accurate approximation of the 1-RDM of a ground state.
In this case, we can construct an active space on $\log(\omega^{-1})$ orbitals which suffices for a constant overlap with the ground state. Since this active space has polynomial dimension, the maximally mixed state on this subspace gives a sufficiently good guiding state for phase estimation.
We summarize these findings with the following result, formally stated and proven as \cref{thm:impurity complexity} in \cref{sec:proof impurity}.

\begin{thm}\label{thm:complexity impurity informal}
  The quantum impurity problem for $H = H_{\free} + H_{\imp}$ is in BQP if either $H_{\free}$ has a constant gap, or if the ground state 1-RDM is given.
\end{thm}

Of course, one generally needs the ground state in order to compute the 1-RDM, so if one requires it beforehand then this does not directly lead to a useful algorithm.
The above result should however be seen as evidence for the usefulness of iterative approaches, where one starts with some estimate of the 1-RDM using computationally inexpensive classical means, and then uses a quantum computer to perform a calculation that leads to a more accurate approximation of the 1-RDM.
See \cite{bauer2016hybrid} for a practical proposal for such a scheme.

The algorithms in \cref{thm:complexity impurity informal} and their analysis do not directly lead to \emph{practically} useful approaches, as their polynomial scaling can be of high degree.
However, these results should be seen as a proof of principle that in quantum impurity problems the `orthogonality catastrophe' is of a relatively mild nature and may not be an insurmountable problem, reducing a potentially exponential scaling to a polynomial one.
In practice, one can likely find better guiding states than the maximally mixed states in the above theoretical algorithms.
Of course, classical algorithms for impurity problems will also perform much better than currently proven by rigorous guarantees \cite{kotliar2006electronic,gull2011continuous,bulla2008numerical}.

Given a large system of $N$ orbitals, where we know that the problem has a strong electron correlation in some subsystem, one can use a quantum embedding method (such as described in \cref{sec:embedding}) to find a reasonable active space $A$ of $K$ modes, and an environment system $C$. One then uses a classical method to find a good guiding state $\ket{\Psi_A}$ on this small active space, and a Slater determinant $\ket{\Psi_C}$ on $C$, and uses $\ket{\Psi_A} \wedge \ket{\Psi_C}$ as a guiding state on the full system.
\cref{thm:Slater det impurity} shows that for a quantum impurity model there \emph{exist} good guiding states of this form, but it does not guarantee that one finds the right subsystem $A$ or state $\ket{\Psi_A}$ in this way.

Finally, we note that the concept of impurity models has also been used to compute properties of strongly correlated materials, as in the Hubbard model. In these cases, a local patch of the system is treated as an impurity, and the interaction with the remainder is treated at a mean-field level \cite{georges1996dynamical}.
It has been proposed to use this framework, with a quantum computer as the impurity solver, for strongly correlated systems \cite{bauer2016hybrid}.
It was demonstrated in \cite{tubmanPostponingOrthogonalityCatastrophe2018} that this leads to a highly multi-configurational problem on the quantum core, but that still, with only 10 determinants, an overlap of around 0.05 can be achieved.

\section{Numerical examples}\label{sec:numerical}

We now proceed to a numerical investigation of two different embedding schemes following the concepts discussed above, namely of a scheme used in the bootstrap embedding (as an example of a density-matrix embedding) and of the Huzinaga (projection) embedding, to assess the influence of quantum embedding on preparation of states with large ground state overlap.
We investigate the relations between the overlap of the mean-field HF basis state, as well as more complex states, with the target state for different spatial sizes of a quantum core and for growing orbital spaces.
As we will see, these different embedding strategies produce rather different orbital bases. 
We study two systems that are exemplary for biomolecules: tryptophan in \cref{sec:tryptophan} (one of the essential amino acids) and a compound containing ruthenium in \cref{sec:ruthenium} with prospects as an anti-cancer drug.
Some of the computational methods are described in more detail in \cref{sec:numerical_methods}.

In general, we observe that in all cases we encounter, the (easy-to-prepare) mean-field state has large overlap with the ground state of the embedded Hamiltonian, with overlaps larger than $0.9$ in almost all cases.
By using more complex states, the overlap can be brought closer to 1.
This does not greatly influence the total cost of quantum phase estimation, but it does allow paralellization with reduced maximal depth (and in many cases this depth reduction will be larger than the additional circuit depth from the more complicated guiding state).

\subsection{Tryptophan}\label{sec:tryptophan}
Some of the twenty essential amino acids, the fundamental building blocks of biomacromolecules, contain unsaturated side chains that are therefore suitable candidates for probing the role of static correlation in biomacromolecules.
Out of these aromatic amino acids, we chose tryptophan as an example.
In addition to this monomeric protein building block, we also study a sequence of tryptophan molecules resembling an oligopeptide structure (\cref{fig:tryptophan}) that can be considered a limiting case for proteins (noting that the primary sequence of amino acids will not necessarily show repetitions composed of only one amino acid).
This situation is used to model a biomolecule of increasing size, so we can see how the overlap of the HF state or sum-of-Slater states decreases with increasing system size (peptide length).
An easy approximation is to consider the tryptophan residues in this series as noninteracting (so the ground state is a product state) so that we can easily estimate the exponential decline of the overlap of the HF determinant and the sum-of-Slater state with peptide length. Switching on the weak interaction only slightly changes the overlap, as we demonstrate for the di- and tripeptides of tryptophan.

\begin{figure}
  \centering
  \includegraphics[width=0.95\linewidth,grid=false]{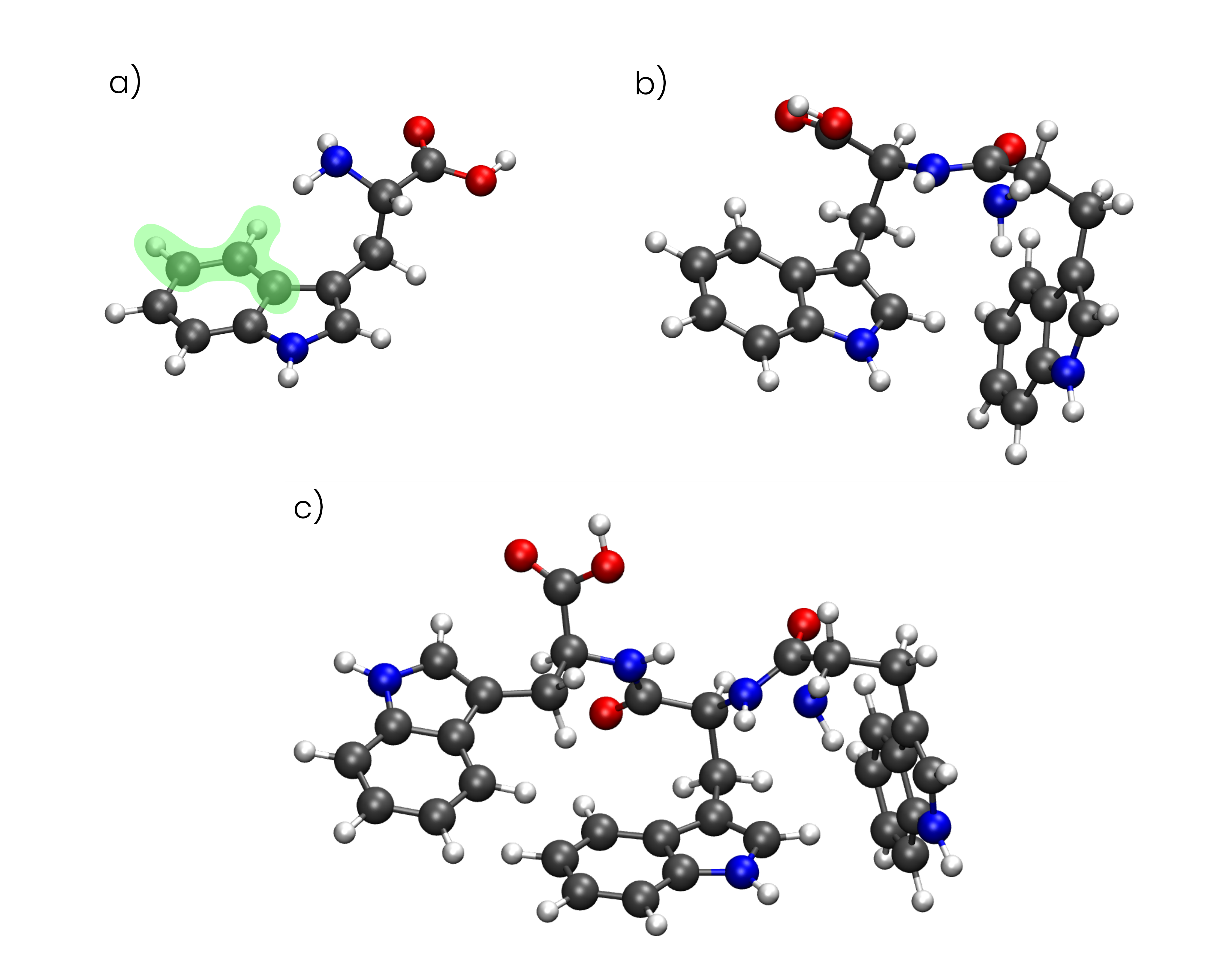}
  \caption{Molecular structure of tryptophan and its derivatives: a) monomer  with the fragment used to compare different embedding approaches indicated in green, b) dimer, and c) trimer. Carbon atoms are indicated in black, hydrogen atoms in white, oxygen atoms in red, and nitrogen atoms in blue.}
  \label{fig:tryptophan}
\end{figure}

Although our results are specific to the examples chosen, we expect their electronic structures to be representative for a more general picture of electronic states. This is due to the regular structure of the amino-acid building blocks of proteins. The results obtained for tryptophan and its oligopeptides will be easily transferable to the homologous and also aromatic amino acid phenyl alanine and may be straightforwardly generalized to the other essential amino acids.

We first applied a density-matrix based embedding. Specifically, we used the bootstrap embedding algorithm for molecular systems \cite{ye2020bootstrap} on a selected five-atom fragment that is part of the aromatic side chain, depicted in \cref{fig:tryptophan}. This fragment was selected as it represents an intuitively challenging embedded system, as it is connected to the environment with both single and delocalized conjugate double bonds. For the bootstrap embedding, a set of orthogonal localized orbitals was constructed using the intrinsic atomic orbitals scheme \cite{Knizia2013}.
This basis was then used to perform the Schmidt decomposition of the HF state and select the set of fragment and entangled bath orbitals. Fragment orbitals selected in this way were localized on atoms, making them a poor choice for the wavefunction representation. Therefore, the selected fragment and bath orbitals were transformed to the eigenvectors of the Fock operator in the fragment-bath space, resulting in canonical orbitals depicted in \cref{fig:orbital_selection_bootstrap} a).
Due to the nature of Schmidt decomposition, these orbitals resemble the canonical HF orbitals of the entire system and, consequently, exhibit relatively small single-orbital entropies (defined as von Neumann entropies of the one-orbital reduced density matrix \cite{legeza2003entropy}). Hence, the multiconfigurational character of the fragment-bath system obtained with bootstrap embedding is comparable to the one of the entire system, even though fragmentation `cuts' through single and double bonds (\cref{fig:orbital_selection_bootstrap} b)). From the drop in the values of the single-orbital entropies one can conclude that there are $12$ orbitals that contribute the most to the multiconfigurational character and the static correlation \cite{stein2016automated,stein2017multiconfigurational_character} and include the reconstructed delocalized aromatic $\pi$ system of the side chain (\cref{fig:orbital_selection_bootstrap}).
Note that in the complete bootstrap embedding, one would apply this construction to multiple fragments, and then perform a matching procedure between the different fragments, a step we ignore for the purpose of this discussion.

\begin{figure}
  \centering
  \includegraphics[width=0.95\linewidth,grid=false]{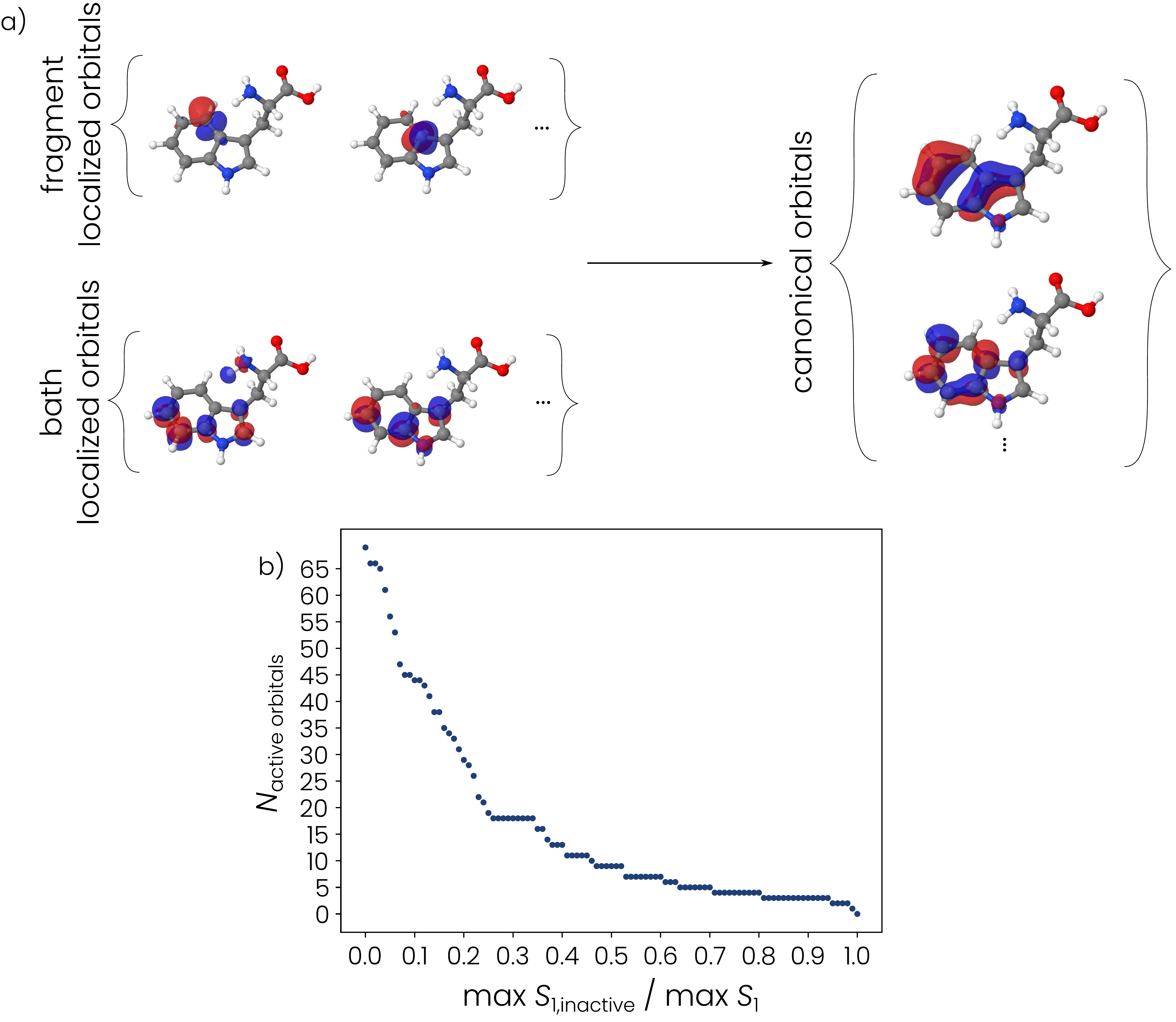}
  \caption{a) Left: fragment and bath orbitals obtained from the Schmidt decomposition of the HF state, Right: canonical orbitals obtained from the HF calculation in the fragment-bath space.
    b) Threshold diagram corresponding to the maximal discarded single-orbital entropies relative to the largest value ($\text{max} \ S_{1}=0.19115$) for different active space sizes, as introduced in Ref.~\cite{stein2016automated}.}
  \label{fig:orbital_selection_bootstrap}
\end{figure}

Next, we performed DMRG calculations
for different active space sizes in the fragment-bath space, with orbitals being selected based on the largest values of the single-orbital entropies, as described in Ref.~\cite{stein2016automated}.
We use a bond dimension of $D=1024$ to yield a reference MPS that can be taken as a reliable approximation to the exact (full configuraion interaction, or FCI) ground state. To assess the suitability of MPS wavefunctions
with smaller bond dimensions for initial state preparation, we calculate both energy differences and overlaps of them with the reference MPS; the results are shown in \cref{fig:bootstrap_overlaps} a) and b).

\begin{figure}
  \centering
  \includegraphics[width=0.95\linewidth,grid=false]{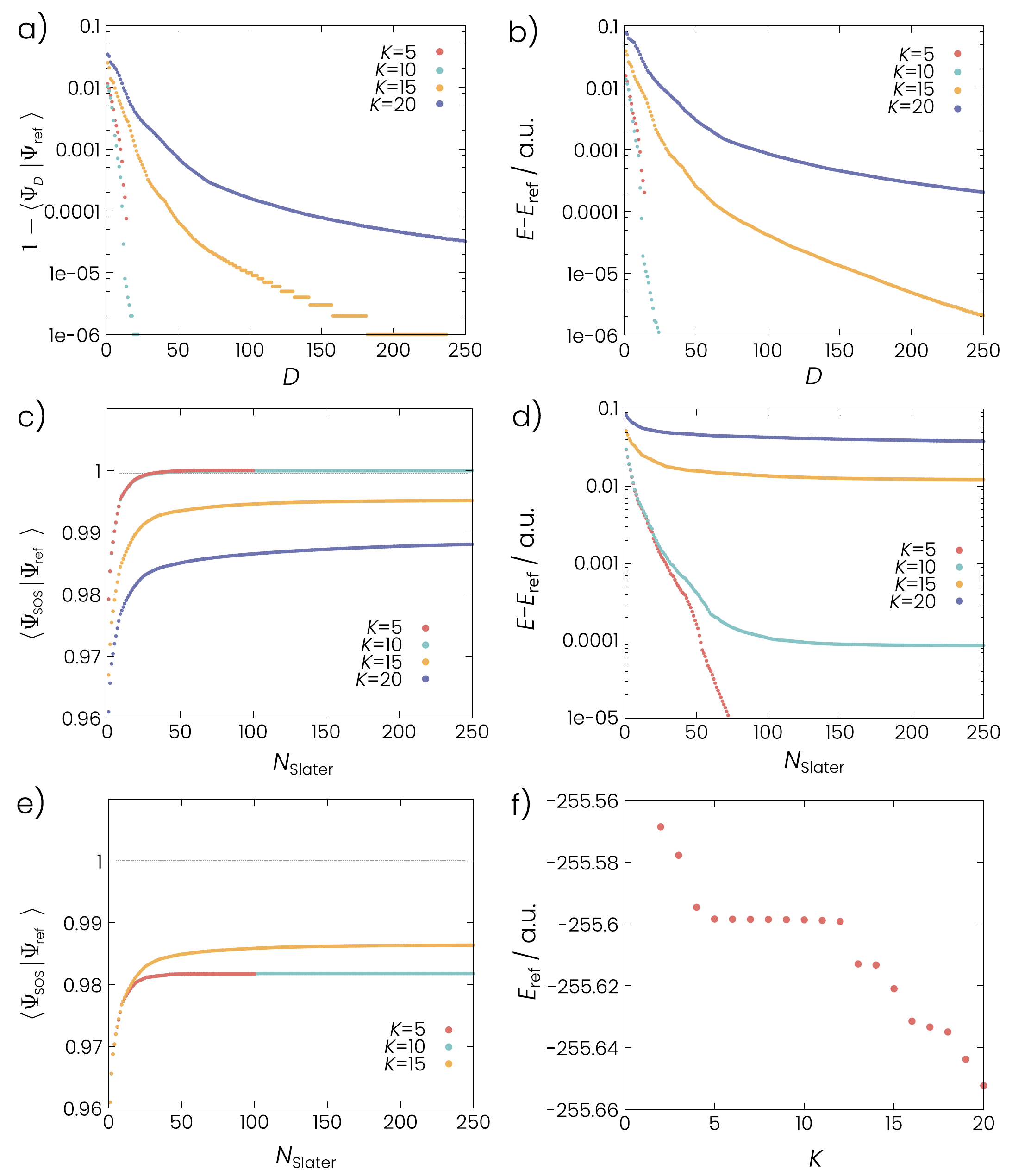}
  \caption{a) Overlap of truncated MPS wavefunctions with the reference MPS of bond dimension $D=1024$ (taken as an approximate FCI result) for different active space sizes $K$. b) Corresponding energy differences of the small-bond-dimension and the reference MPS. c) Overlap of sum-of-Slater states with different numbers of Slater determinants with the reference MPS wavefunction d) Corresponding energy differences of the sum-of-Slater states and reference MPS state e) Overlap of sum-of-Slater states obtained from reference MPS wavefunctions in smaller active orbital spaces with the reference MPS wavefunction of the largest active space considered ($K=20$) f) Reference MPS energies for different active space sizes}
  \label{fig:bootstrap_overlaps}
\end{figure}

From \cref{fig:bootstrap_overlaps} one can see that the overlap of the truncated states with the reference MPS state approaches a value of $1$ relatively quickly. This is expected, as the truncation of the bond dimension is performed in the optimal way using the singular value decomposition, selecting the state of the lower bond dimension that has maximal overlap with the nontruncated one. For similar reasons, the energies of the small-bond-dimension states quickly approach the corresponding reference energy. Furthermore, even states with bond dimension $D=2$ already demonstrate excellent overlap with the reference MPS state ($\approx 0.97$), making such an MPS a good candidate for state preparation.

We also consider a different ansatz, namely sum-of-Slater states (see \cref{sec:orthogonality} and \cite{fomichevInitialStatePreparation2023} for a discussion of such states on quantum computers).
We obtain these by selecting the most important determinants from the FCI-type expansion represented by the reference MPS wavefunction, keeping the $N_{\text{Slater}} = L$ terms with the largest amplitudes:
\begin{align}\label{eq:sos}
  \ket{\Psi_{\text{SOS}}} = \frac{1}{\sqrt{\sum_{i=1}^L \abs{C_i}^2}} \sum_{i=1}^L C_i \ket{\Phi_i},
\end{align}
where $\ket{\Phi_i}$ are Slater determinants in a fixed basis of orbitals, and we have ordered so $\abs{C_1} \geq \abs{C_2} \geq \dots$.
These states are both useful as a guiding state Ansatz (for small $L$), and to understand the correlation structure of the ground state.
From \cref{fig:bootstrap_overlaps} c), it can clearly be seen that the overlap obtained for these sum-of-Slater expansions with the reference MPS wavefunction are again relatively large.
However, convergence of these overlaps with the number of included determinants is much slower than the convergence of small-bond-dimension MPS wavefunctions with increasing bond dimension.
Moreover, convergence is slowed down even further for increasing active space sizes $K$, which may be taken as an indication of the orthogonality catastrophe in the large active space limit. It is thus necessary to weigh both the gate count for the preparation of MPS and sum-of-Slater states and the quality of the overlap that they provide for the specific system of interest.

With respect to the correlation structure of the ground state, it is noteworthy that the change of overlap exhibits two different regimes, with a large jump at a small number of Slater determinants, followed by a slow convergence to unity. The initial jump corresponds to the inclusion of the Slater determinants describing excitations into the virtual orbitals with the largest values of the single-orbital entropies, which carry the largest coefficients in the FCI expansion and account for static correlation. Although overlaps with the reference MPS wavefunction are comparatively large, energies corresponding to the sum-of-Slater states converge slowly to the reference MPS energy (\cref{fig:bootstrap_overlaps} d)). Such behavior is expected, as truncation of a Slater-determinant expansion results in a neglect of a significant portion of the dynamical correlation energy. By contrast, an MPS, even if truncated to very low bond dimension, can represent a large number of Slater determinants, resulting in much better convergence behavior, which makes them a good candidate for guiding states with high ground state overlap \cite{fomichevInitialStatePreparation2023}.

We now investigate the quality of the sum-of-Slater states obtained for smaller active spaces as initial states for the larger active space ($K=20$). As can be seen in \cref{fig:bootstrap_overlaps} e), the sum-of-Slater ansatz prepared in an active space with $K=15$ orbitals provides overlaps that are virtually indistinguishable from the reference. In the case of $K=5$ and $K=10$, the overlap is somewhat lower
when compared to the sum-of-Slater state prepared in the $K=15$ case, but still represents an improvement over the Hartree-Fock state. This is due to the fact that not all highly entangled orbitals are included in these cases (as can be seen from the \cref{fig:orbital_selection_bootstrap} b)), resulting in a lack of several determinants with non-negligible coefficients. However, it can be seen that as long as the most entangled orbitals are included, these active spaces provide sum-of-Slater states that are of the same quality in terms of overlap as the ones obtained from the calculation on the larger active space.

Finally, from the \cref{fig:bootstrap_overlaps} f), it is evident that many orbitals must be considered to reach chemical accuracy in absolute energies, as they carry a significant portion of the correlation energy. The fact that the initial state can be prepared in smaller active spaces therefore significantly reduces the computational overhead for the classical part of the initial state preparation.

\begin{figure}
  \centering
  \includegraphics[width=0.9\linewidth,grid=false]{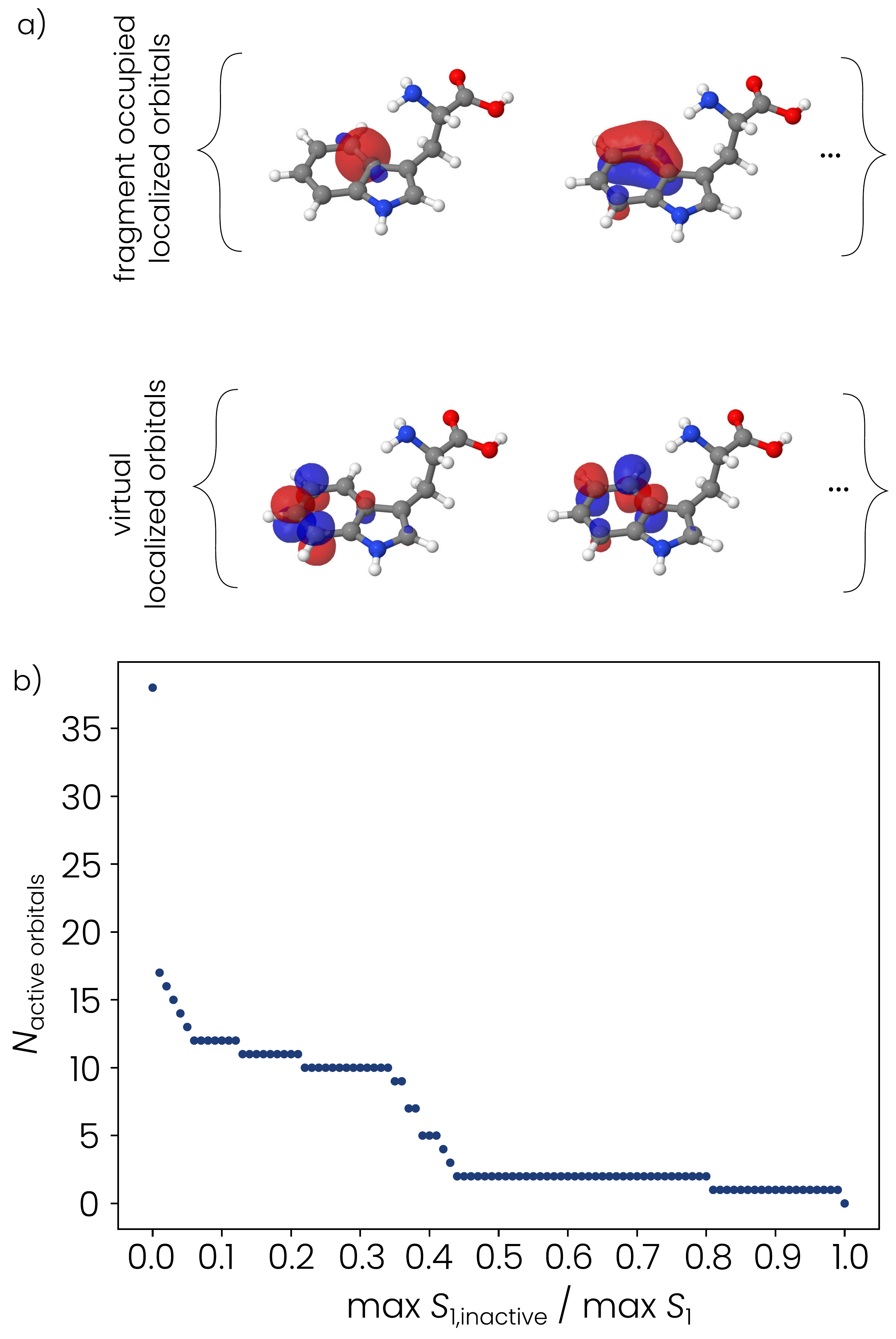}
  \caption{a) Localized occupied fragment orbitals and localized virtual orbitals obtained in Huzinaga embedding b) Threshold diagram
    corresponding to the maximal discarded single-orbital entropies relative to the largest value ($\text{max} \ S_{1}=0.21866$) for different active space sizes, as introduced in Ref.~\cite{stein2016automated}. This may be contrasted with \cref{fig:orbital_selection_bootstrap}.}
  \label{fig:orbital_selection_DFT}
\end{figure}

We also investigated a projection-based Huzinaga embedding \cite{hegely2016exact, Chulhai2018}, using the same five-atom fragment, in order to assess how different embedding strategies affect the state preparation problem. There are several key differences with respect to bootstrap embedding.
Firstly, the orbitals used for the Huzinaga embedding are obtained from a DFT calculation on the entire molecule, while the basis for the density-matrix approach in the bootstrap embedding is a HF calculation.
This results in a different Hilbert space for the fragment problem.
Secondly, in the Huzinaga embedding, interaction of the fragment orbitals with the occupied environment is described with a DFT potential, which approximates the correlation contributions from the occupied environment orbitals that are neglected in the case of bootstrap embedding. This interaction modifies the one-electron part of the fragment Hamiltonian.
Third, a split orbital localization scheme is applied within the Huzinaga embedding, based on intrinsic bond orbitals \cite{Knizia2013,Senjean2021}, while canonical orbitals are used in the bootstrap embedding, which results in a qualitatively different multiconfigurational character of the fragment wavefunction in these cases.
Unlike bootstrap embedding, the entire virtual space is considered in Huzinaga embedding and the canonical orbitals are not recomputed after the embedding. Due to the local nature of the orbitals, the values of the single-orbital entropies are significantly different from the bootstrap embedding case
(\cref{fig:orbital_selection_DFT} and \cref{fig:orbital_selection_bootstrap}). In Huzinaga embedding, orbital entropies are larger when compared to the bootstrap case, which can be attributed to the usage of localized orbitals \cite{weser2021entanglement_local_vs_delocalized}.
This results in a larger number of highly entangled orbitals, while the contribution of the virtuals that are spatially separated from the embedded fragment decreases significantly.

Due to the small entanglement of some virtual orbitals, convergence of the sum-of-Slater states both in terms of overlap and energy is faster than in the bootstrap embedding case (\cref{fig:figure_DFT_overlaps} a) and b)). Hence, localized orbitals might present a more suitable basis for initial state preparation.
In classical quantum chemical methods one aims for a pronounced single-configurational character of the wavefunction and accounts for dynamic correlation a posteriori.
In contrast, for quantum computers, it may be more desirable to have a multiconfigurational wavefunction with little residual dynamic correlation (\cref{fig:figure_DFT_overlaps} d)), since given a guiding state with only constant overlap, phase estimation gives accurate energies on the fragment orbitals.
Similar to our bootstrap embedding results, initial states with large overlap can be prepared in a smaller active space, reducing the cost of a classical calculation for the state preparation step also in the case of Huzinaga embedding (\cref{fig:figure_DFT_overlaps} c)).

\begin{figure}
  \centering
  \includegraphics[width=0.9\linewidth,grid=false]{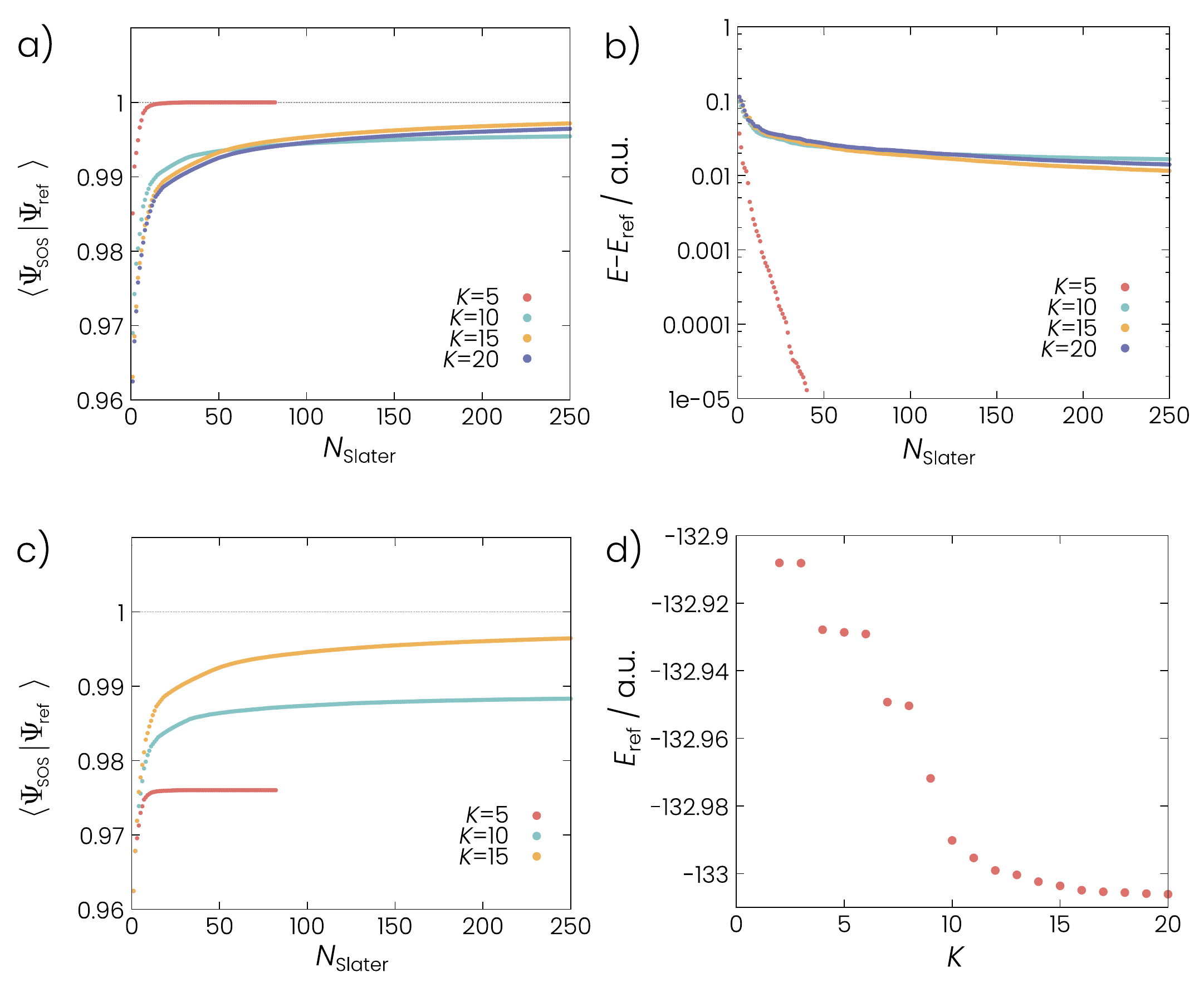}
  \caption{a) Overlap of sum-of-Slater states (with $N_\text{Slater}$ Slater determinants corresponding to the largest coefficients) with the reference MPS state of bond dimension $D=1024$ for different active space sizes $K$. b) Overlap of sum-of-Slater states
    for different active space sizes $K$ with the reference MPS state ($D=1024$) for the largest active space considered $K=20$. c) Energies of sum-of-Slater states, with lines corresponding to the energies of the full MPS wavefunction. d) Energy of the MPS ground state energy approximation, as a function of active space size $K$.}
  \label{fig:figure_DFT_overlaps}
\end{figure}

Finally, we turn to the sequence of tryptophan molecules mimicking an oligopeptide structure. We chose the active spaces to be residing entirely on the side chain (see \cref{sec:numerical_methods} for details). First, we consider the limit of noninteracting side chains corresponding to separated tryptophan molecules in the sequence. In this case, the overlaps of the HF state with the reference MPS states
decay exponentially with the number of tryptophan molecules (\cref{fig:overlaps_polytryptophan} a)).
In the case of sum-of-Slater states, the decay is approximately exponential as well. The exact exponential dependence occurs when the Slater determinants included correspond to the wavefunction that is a product of sum-of-Slater states on each monomer.
Since the monomer ground state has relatively high overlap with the HF state, we see that the overlap of the HF state is still substantial (at around 0.7) for a sequence of 5 tryptophane molecules.

We also compute the overlap with the ground state of the \emph{interacting} tryptophan dimer and trimer.
It can be seen that weak interactions (between the side chains) only marginally perturb the overlap behavior observed for the separately treated side chains, so the exponential decay from the non-interacting case is a good approximation.
We performed embedding calculations for a single side chain in ditryptophan (at the C-terminus side chain) and tritryptophan (side chain at the amino acid in the middle) and investigated the resulting overlaps of the sum-of-Slater states. As can be seen in \cref{fig:overlaps_polytryptophan} b), the behavior of the overlap in both cases of the embedded monomers is essentially the same as in the case of a single tryptophan side chain.
From these observations, we may conclude for proteins that embedding of relevant amino acid side chains represents a good strategy for mitigation of the orthogonality catastrophe caused by large system size.

\begin{figure}
  \centering
  \includegraphics[width=0.95\linewidth,grid=false]{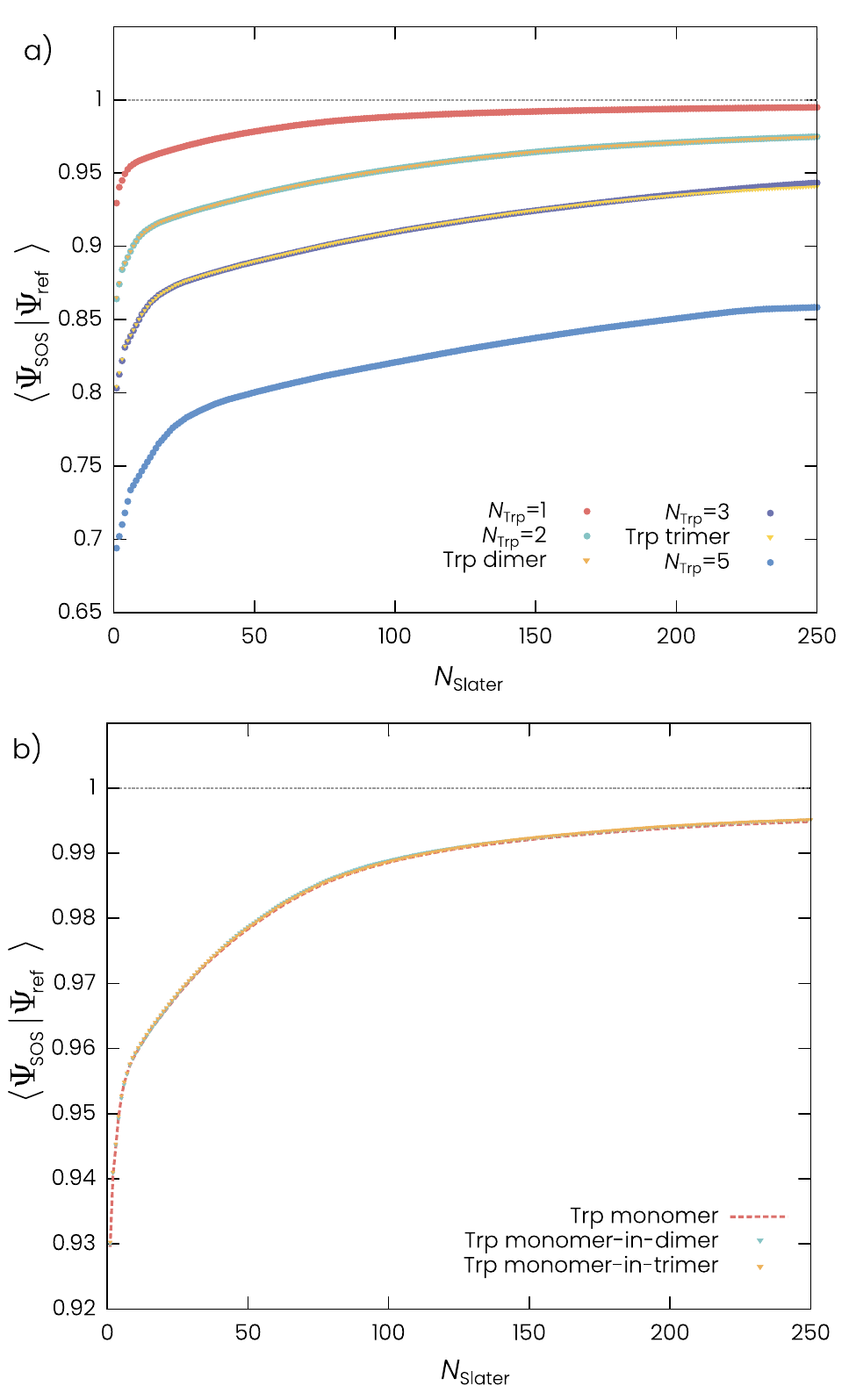}
  \caption{a) Overlap of sum-of-Slater states containing $N_\text{Slater}$ Slater determinants with the reference MPS state of bond dimension $D=1024$
    for different numbers of noninteracting tryptophan residues (circles) and interacting residues in ditryptophan and tritryptophan (triangles, note: the data for the interacting dimer and trimer overlaps with the data for the noninteracting cases) b) Same as a) for a single monomer embedded into the ditryptophan
    and tritryptophan (triangles). Values for the isolated monomer are given as a dashed line.}
  \label{fig:overlaps_polytryptophan}
\end{figure}

\subsection{Ruthenium anti-cancer drug}\label{sec:ruthenium}
As a second example, we consider a system containing elements beyond the second period of the periodic table. This compound (see Fig. \ref{fig:ruthenium} for a ball-and-stick representation of its structure) is an anti-cancer drug \cite{Trondl2014, Peti1999}
which can bind as an inhibitor to a glucose-regulating protein \cite{Macias2011}.
Such a binding is a typical example of small-molecule drug recognition by biomacromolecules and therefore an example for a
molecular recognition application.
Here, we focus on the isolated Ru-based drug molecule and consider its central region, the Ru ion containing moiety, as a quantum core (as highlighted in green in Fig. \ref{fig:ruthenium}).
We leave an investigation of how the ground state overlap problem changes when including the target protein in the embedding to future work.

\begin{figure}
  \centering
  \includegraphics[width=0.8\linewidth,grid=false]{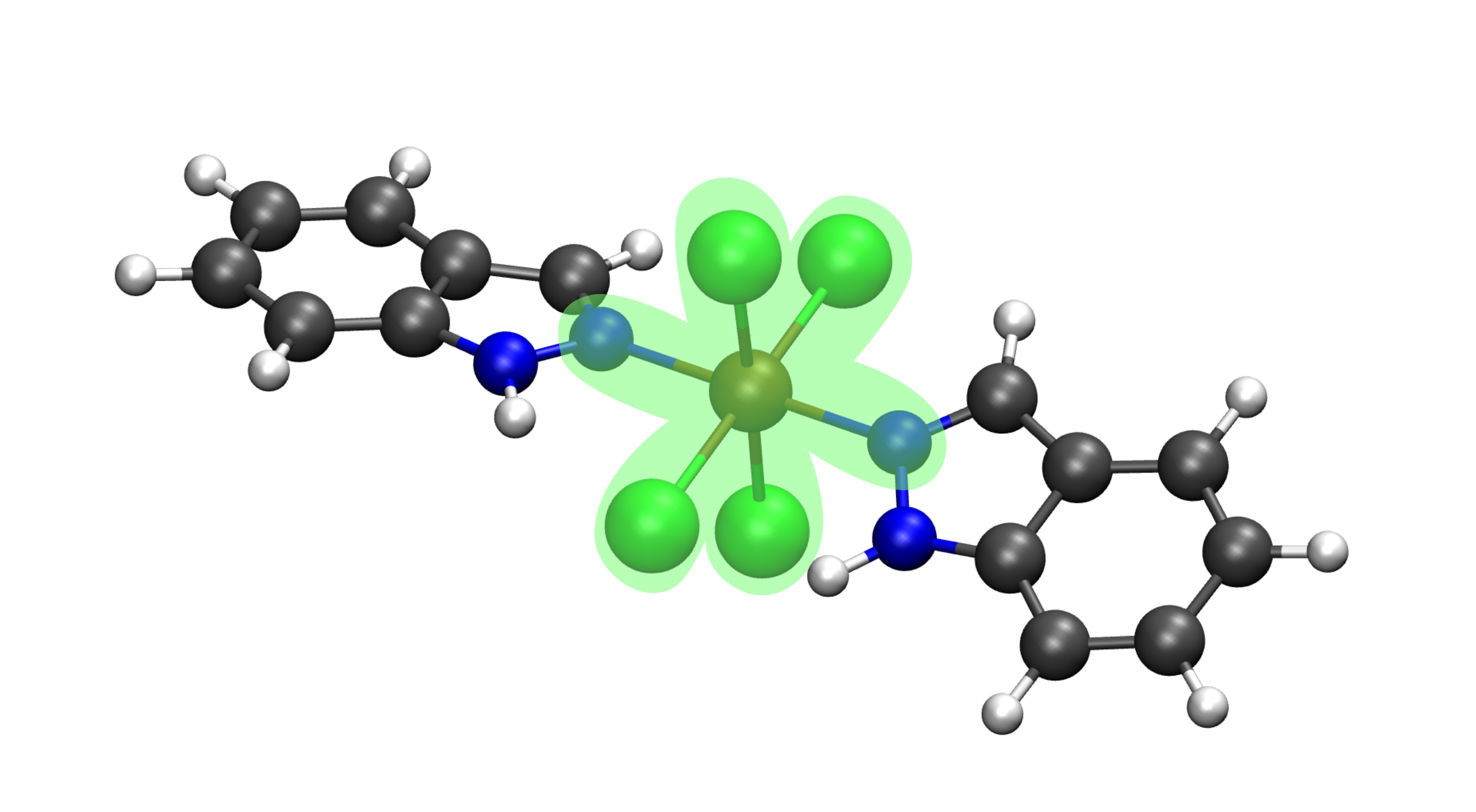}
  \caption{Molecular structure of the ruthenium complex with the fragment used for the Huzinaga embedding indicated in transparent green. Carbon atoms are indicated in black, hydrogen atoms in white, nitrogen atoms in blue, chlorine atoms in green, and the ruthenium atom in brown.}
  \label{fig:ruthenium}
\end{figure}

For this complex, we applied the Huzinaga embedding of the ruthenium ion and its first (nearest-neighbor) coordination sphere (\cref{fig:ruthenium}). We consider two different
charges of this complex: the anion, $q=-1$, corresponds to a doublet state, and the neutral complex, $q=0$, we considered as a triplet ground state. For the doublet, the overlap of the HF state and the sum-of-Slater states is very large, indicating a single-configurational character of the wavefunction. By contrast, the triplet state exhibits smaller overlap of the HF determinant with the reference MPS state of bond dimension $D=1024$,
which is quickly cured by including a second Slater determinant to yield a large overlap, see \cref{fig:overlaps_ruthenium}. This behavior is a consequence of the triplet nature of the wavefunction, which cannot be described by a single spin-restricted determinant. However, it can be described with a configuration state function which is a symmetry-adjusted basis state. In the case of the triplet, the configuration state function is a linear combination of two restricted Slater determinants (and therefor multi-configurational). If the system of interest is comprised of several high-spin regions (e.g., in the case of metal clusters), the overlap of the HF determinant will further decrease,
depending on the number and spins of these regions \cite{leeEvaluatingEvidenceExponential2023}.
It has already been argued that spin coupling to produce configuration state functions effectively solves this problem \cite{marti2024spin}.
However, we note that such situations are not at all common in biomolecular recognition situations.

\begin{figure}
  \centering
  \includegraphics[width=0.95\linewidth,grid=false]{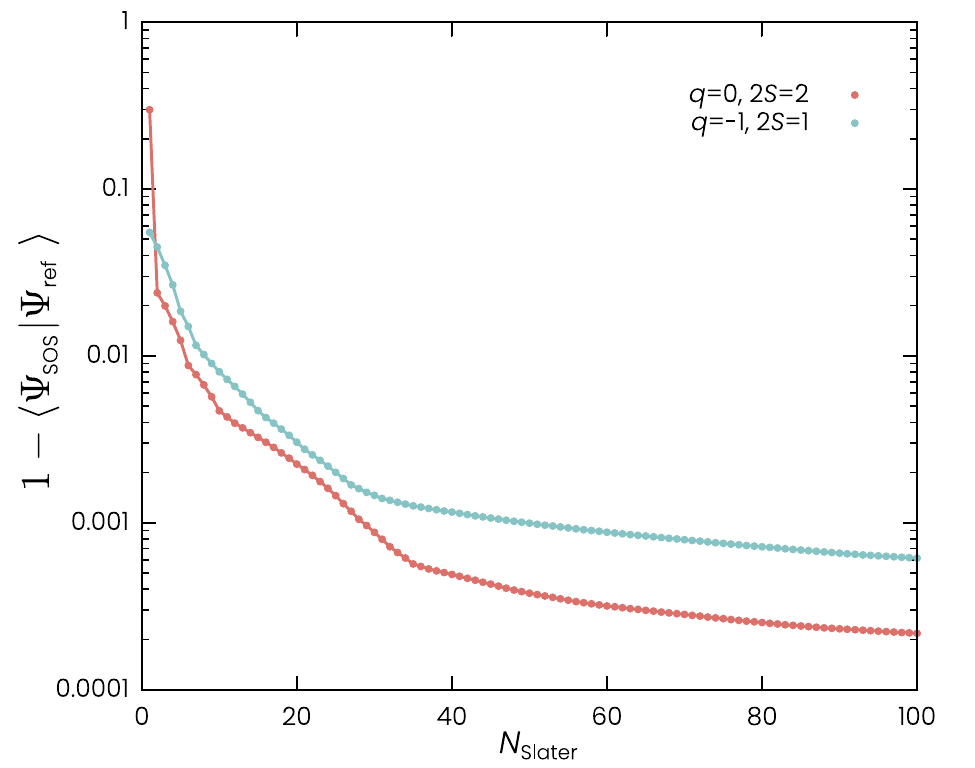}
  \caption{Overlap of sum-of-Slater states (with $N_\text{Slater}$ Slater determinants corresponding to the largest coefficients) with a reference MPS state of
    bond dimension $D=1024$ for the doublet (net charge $q=-1$) and triplet (net charge $q=0$) electronic states of the embedded fragment in the ruthenium complex.}
  \label{fig:overlaps_ruthenium}
\end{figure}

\section{Conclusions}
In this work, we have demonstrated that embedding methods effectively avoid the orthogonality catastrophe in practice for problems with local spatial structures, specifically in biochemistry.
We show that for two conceptually different quantum embedding methods, low complexity states serve as good guiding states for quantum phase estimation.
We demonstrate that this is the case for the fundamental building blocks of biomolecules, but also for other compounds such as small metal-containing drug molecules.
As expected, the ground state overlap of the mean-field HF basis state with the target state decreases with increasing active space size, but based on our estimates, phase estimation will likely face bottlenecks due to the polynomial scaling of simulation methods with the number of orbitals $N$ before small or vanishing ground state overlaps become problematic.
If we consider embedded fragments with a fixed number of electrons and an increasing number of orbitals, in principle FCI on a classical computer scales polynomially as $\mathcal{O}(N^n)$. This means that the speed-up provided by quantum phase estimation is polynomial, as the Hamiltonian simulation subroutine has polynomial complexity in $N$, with $\mathcal{O}(N^4)$ scaling even without any truncation and factorization strategies. However, for a modest number of electrons, on the order of $\sim 20$ for exact solvers and on the order of up to $\sim 100$ for approximate solvers such as DMRG, FCI becomes unfeasible as a classical method. Therefore, even though the quantum advantage of phase estimation algorithms in this case is in principle only polynomial, it is nevertheless significant. Additionally, with increasing quantum resources we can also increase the number of electrons in the fragment in order to reduce the edge effects of the embedding. This will restore the exponential scaling size of the Hilbert space.

For general applicability, we have also studied embedding and guiding state preparation for impurity models. Here we argued for the existence of accurate small active orbital spaces and for mild scaling of ground state overlap. This is clearly a positive prospect for quantum algorithms for ground state energy estimation, especially for biochemical applications. Additionally, the quantum description of a small reactive region in a macromolecular structure will always be sufficient for the investigation of relevant chemical processes as these can be considered to occur locally in a molecular structure (bond-breaking/forming involves a limited number of atoms). Therefore, focusing on the embedded region only is expected to be sufficient for such cases. If the quantum description of a larger region is required, this will also be possible. For instance, bootstrap embedding can be utilized to reconstruct the total energy from calculations on many smaller, overlapping embedded fragments \cite{weisburn2024embedding} or stitch local energies for (small) embedded fragments together with a regression model or a machine learning potential.
We emphasize that these procedures have already been established within traditional computing and the insight is that they alleviate the orthogonality catastrophe and allow for efficient state preparation required for quantum phase estimation to obtain energies with controlled accuracy eventually.

To provide a more general framework for formal analysis, here we considered also embedding and guiding state preparation for impurity models. In this context, we could argue in favor of the existence of accurate small active orbital spaces and for mild scaling of ground state overlap. We further demonstrated that orbital selection based on the quantum information theory and single-orbital entropies provides such a small active space in which state preparation can be performed. An interesting open question is whether the quantum impurity problem without any further restrictions is in BQP.
The mathematical analysis of embedding methods and of quantum impurity problems is based on the natural occupation numbers of the 1-RDM. Future work could explore whether the structure of 1-RDMs, which satisfy additional linear constraints \cite{altunbulak2008pauli,schilling2018generalized}, can be used to improve quantum embedding schemes \cite{theophilou2015generalized,schilling2020implications,faulstich2022pure} and help to find good guiding states for phase estimation.

Various interesting challenges remain to assess the utility of fault-tolerant quantum computers for ground state energy estimation of electronic structure in chemistry.
As we have emphasized, one of the main advantages of quantum computing based on phase estimation over conventional methods is its accuracy guarantees. However, embedding methods introduce an (uncontrolled) error, even if small in practice. A better understanding of this error and the effect of the approximate treatment of the environment on the quantum core will be important for applications of quantum computers to macromolecules, and it requires further work.

\section*{Acknowledgements}
This work is part of the Research Project ``Molecular Recognition from Quantum Computing''. Work on this project is supported by Wellcome Leap as part of the Quantum for Bio (Q4Bio) Program.
We also acknowledge financial support from the Novo Nordisk Foundation (Grant No. NNF20OC0059939 ‘Quantum for Life’).

\appendix

\section{Guiding states in quantum chemistry}\label{sec:orthogonality}
This appendix contains a review of the guiding state problem for quantum chemistry.
We start by briefly reviewing the dependence of quantum phase estimation on the ground state overlap of the guiding state.
Next we give an overview of different ansatzes for guiding states.
We discuss different types of correlation in quantum chemistry, and why one expects that the ground state overlap vanishes exponentially for increasing system size (the orthogonality catastrophe).

\subsection{Quantum circuits and guiding state preparation}

To estimate the ground state energy, one would like to apply quantum phase estimation with the time-evolution operator $U(t) = \exp(itH)$.
Then, using a guiding state $\ket{\psi}$ with ground state overlap $\abs{\braket{\psi}{\psi_0}} \geq \overlap$, one finds an approximation to the ground state energy with probability $\overlap^{2}$.
Therefore, the ground state overlap is related to the number of repetitions needed to find the ground state energy through phase estimation.
An approximation of accuracy $\eps$ requires a circuit with time evolution for time $\bigO(\eps^{-1})$, leading to a total cost of $\tilde \bigO(\eps^{-1} \overlap^{-2})$.
The scaling with $\eps^{-1}$ is known as Heisenberg scaling and is optimal \cite{giovannetti2006quantum}, but different approaches to phase estimation have different scalings with the ground state overlap.
Given access to a quantum circuit preparing $\ket{\psi}$, one can find the ground state energy using $\bigO(\overlap^{-1})$ uses of this state preparation unitary, and time evolution for time $\tilde\bigO(\eps^{-1} \overlap^{-1})$.
This scaling is optimal \cite{mande2023tight} but does, however, lead to deeper circuits.
On the other hand, there are also approaches that only require a single ancilla qubit and time evolution for time $\tilde \bigO(\eps^{-1})$ \cite{lin2022heisenberg}; these require $\tilde\bigO(\overlap^{-4})$ circuit repetitions.
Finally, in the regime where $\overlap \to 1$ it is possible to reduce the maximal circuit depth, at the cost of an increased number of repetitions.
If $\delta = 1 - \overlap^2$, the maximal required time evolution can be reduced to $\bigO(\delta \eps^{-1})$, with a total evolution time of $\bigO(\delta \eps^{-2})$ \cite{ding2023even,ni2023low}.
This means that if one can find guiding states with ground state overlap surpassing $\overlap \geq 0.9$ using cheap conventional methods, this may be used to signifificantly reduce the required circuit depth for accurate energy estimates using quantum phase estimation, which may be helpful for realizing such simulations on devices with a limited number of qubits and maximal circuit depth.

There are three standard classes of guiding states for chemistry problems, which we will now briefly review, together with the cost of preparing them on a quantum computer. We refer to \cite{fomichevInitialStatePreparation2023} for an extensive overview.
For convenience we assume we are using a Jordan-Wigner mapping for the fermion-to-qubit mapping of the electronic structure problem.
We quote the best known gate counts for two measures: the number of two-qubit gates required (and ignoring the number of single-qubit gates); or alternatively the number of Toffoli gates (and ignoring the number of Clifford gates). The first measure is relevant for current and near-term devices where two-qubit gates are typically much slower and more noisy than single-qubit gates. The second measure applies to a fault-tolerant model where the main cost comes from the non-transversal (non-Clifford) operations.

\begin{enumerate}
  \item \textbf{Single Slater determinant:} This is the most basic case. Hartree-Fock yields a single Slater determinant, and under a fermion-to-qubit mapping this state is represented by a product state (which requires no two-qubit gates or $T$ gates to prepare). If we use the standard Jordan-Wigner mapping with canonical orbitals (compatible with the HF state), then for $N$ spin-orbitals and $n$ electrons this is simply the state $\ket{1}^{n} \ket{0}^{N-n}$.
        It is also possible to choose a single Slater determinant in a different orbital basis than the one used for the fermion-to-qubit mapping, which can significantly improve the ground state overlap \cite{tubmanPostponingOrthogonalityCatastrophe2018,ollitrault2024enhancing}.
        A single-particle basis change can be implemented as a quantum circuit using so-called Givens rotations. Using the Jordan-Wigner representation, this requires $n(N-n)$ two-qubit gates \cite{kivlichan2018quantum}.
  \item \textbf{Sum-of-Slater determinants:} A direct extension is to take a state that is a superposition of a (small) number $L$ of Slater determinants, as in \cref{eq:sos}. Such states can for example be found through configuration interaction with single and double excitations (CISD) methods or selective configuration interaction (SCI) methods \cite{fomichevInitialStatePreparation2023}. Alternatively, they can be extracted as the dominant coefficients of a matrix product state (see below).
        Again, if we assume that the fermion-to-qubit mapping uses the same orbital basis, such states are mapped to a superposition of $L$ standard basis states. These can be prepared using at most $\bigO(NL)$ two-qubit gates, see \cite{malvetti2021quantum} for explicit gate counts; or $\bigO(L \log(L))$ Toffoli gates and $\big(\log(L))$ ancilla qubits \cite{fomichevInitialStatePreparation2023}.
        Additionally, if the number of excitations from the HF state is bounded by $k$, $\bigO(Lk)$ two-qubit gates suffice \cite{deveras2022double}.
  \item \textbf{Matrix product states:} A final power class of states are matrix product states (MPS). These can be found using the DMRG algorithm, and they are characterized by their bond dimension $D$.
        Accurate quantum chemistry calculations typically require a large bond dimension, but alternatively one can either try to minimize over low bond dimension MPS, or truncate a high bond dimension MPS to a low bond dimension state which has less accurate energy but still significant ground state energy overlap. In many cases, this leads to high-quality guiding states \cite{fomichevInitialStatePreparation2023}, but at a relatively high classical processing cost.
        MPS can be prepared in sequential fashion, using $\bigO(ND^2)$ gates (either two-qubit or Toffoli) \cite{schon2005sequential} and depth scaling with $N$.
        Heuristic methods for short depth circuits preparing states with high overlap with a given MPS can be found in \cite{ran2020encoding,rudolph2023decomposition}.
        Other depth reduction techniques are based on the correlation length of the MPS \cite{malz2024preparation}, or by using mid-circuit measurements and adaptive circuits \cite{smith2024constant}.
\end{enumerate}

Other methods for finding and preparing guiding states, not discussed in this work, include adiabatic state preparation \cite{veis2014adiabatic,albash2018adiabatic}, thermal state preparation algorithms \cite{chen2023quantum} and variational approaches based on unitary coupled cluster methods \cite{lee2018generalized}.

\subsection{The orthogonality catastrophe}\label{sec:orthogonality_catastrophe}
In most cases in many-body physics, requiring approximations at the level of the wavefunction that have large overlap with the true ground state is a very strong requirement. As we will explain below in more detail, small (local) errors in the approximation lead to essentially orthogonal states on the many-body level, a phenomenon known as the orthogonality catastrophe.
Nevertheless, for quantum algorithms it is in fact important to prepare guiding states that have a significant overlap with the true many-body ground state in order to guarantee that the quantum phase estimation algorithm will give a good estimate of the ground state energy.

We collect four basic arguments for the orthogonality catastrophe, based on \emph{accumulation of error} and \emph{Andersons impurity argument} in the thermodynamic limit, the \emph{electron-electron cusp} in the continuum limit, and finally an argument from the \emph{computational complexity of the ground state problem}.

The most basic intuition for the orthogonality catastrophe is that if we have an extended system consisting of $N$ subsystems without any correlation between the systems, the global ground state overlap decays exponentially if there is a local approximation error.
If the true state is $\ket{\phi}^{\ot n}$, and we have on-site estimates $\ket{\psi}$ of $\ket{\phi}$, then the global squared overlap decays exponentially as $\abs{\braket{\phi}{\psi}}^{2N}$ (and the same is true for uncorrelated fermionic systems) \cite{kohn1999nobel}.
While this gives a reasonable intuition, it is also a very artificial scenario.

In \cite{anderson1967infrared} Anderson showed that a similar phenomenon occurs for a realistic physical system of an electron gas with an impurity. The impurity is a local sized perturbation, which is constant-sized compared to the number of electrons. The ground state of the perturbed system has an overlap with the ground state of the unperturbed system which, in this case, decays \emph{polynomially} with the system size. Note that this differs from the uncorrelated example, where the decay arises from the fact that \emph{every} local fragment incurs an error.
This phenomenon is known as the Anderson orthogonality catastrophe.
This can be derived in perturbation theory for a one-body perturbation \cite{anderson1967infrared}.
This means that for impurity models, the orthogonality `catastrophe' may only be of a polynomial nature and therefore need not be an obstruction to an efficient quantum algorithm as one may hope that the ground state has a polynomially decaying overlap with the mean-field state.
We partially make this idea rigorous in \cref{sec:impurity} and \cref{sec:proof impurity}, based on the work of \cite{bravyi2017complexity}.

Another fundamental reason for the orthogonality catastrophe lies in the continuum limit rather than the thermodynamic limit.
Here we keep the number of electrons $n$ fixed, but we increase the spatial resolution of the second-quantized Hamiltonian by raising the number of orbitals $N$. This leads to convergence to the true wavefunction in $L^2(\RR^{3n}) \ot \CC^2$. Note that the Hilbert space of a fixed number of electrons $n$ but increasing $N$ scales polynomially (but with an exponent scaling with $n$).
It is well known that resolving the electron-electron cusp of the wavefunction of any system with Coulomb interactions requires a diverging number of Slater determinants (see for example the discussion in Chapter 7 of \cite{helgakerMolecularElectronicStructureTheory2014}).
This is well understood analytically \cite{sobolev2022eigenvalue,hearnshaw2022analyticity,hearnshaw2023diagonal}, and is true for arbitrary eigenfunctions of the electronic Hamiltonian.
In particular, it is known that the 1-RDM in the continuum limit has eigenvalues (i.e. natural occupation numbers) decaying as $\lambda_k \sim k^{-\frac{8}{3}}$ as $k$ goes to infinity \cite{sobolev2022eigenvalue}, lower bounding the required number of (natural) orbitals needed to approximate the exact wavefunction.

A final argument for a version of the orthogonality catastrophe comes from complexity theory.
The problem of computing ground state energies is known to be QMA-complete, which means that it is strongly believed to be hard for quantum computers.
This remains true in the case of electronic Hamiltonians \cite{schuch2009computational,ogorman2022intractability}. As a result, there should be no efficient (quantum) algorithm for finding states with at most polynomially decaying overlap with the true ground state: if such an algorithm existed, it could be used to prepare an initial state for the quantum phase estimation algorithm and thus efficiently estimate the ground state energy.
On the other hand, it is known that the ground state estimation problem remains BQP-hard even if one is provided a guiding state with very high ground state overlap \cite{cade2023improved}.
This means that in general, having classical methods which achieve reasonable overlap does not imply the ability to improve this precision arbitrarily in an efficient manner (although this could still be the case in practice for problems in chemistry, as suggested by \cite{leeEvaluatingEvidenceExponential2023}).

Finally, for the general class of Hamiltonians that have an (unknown) efficiently describable guiding state whose expectation values can be computed efficiently using a classical computer (such as MPS), the ground state energy problem is QCMA complete \cite{weggemans2023guidable}. This provides evidence that in general guiding states that can efficiently be described classically may still be hard to find.

\subsection{Guiding states in chemistry}
Electronic structure theory is the main quantum mechanical problem in chemical and materials science.
We distinguish three regimes:
\begin{enumerate}
  \item \textbf{The weak correlation limit:} In the weak correlation limit the mean-field Hartree-Fock state is already an excellent qualitative approximation. As discussed, these are easy to prepare on a quantum computer, and this suffices as the choice of guiding state.
        Most chemical processes belong to this class. However, for these systems coupled cluster models based on a Hartree-Fock reference state deliver reliable results because they can efficiently account for the lacking dynamic correlation (yet, with unknown system-focused error for a specific application \cite{reiher2022molecule}). Small molecules \cite{tubmanPostponingOrthogonalityCatastrophe2018} belong to this class, as well as electronic host-guest binding energy calculations such as those required in drug design problems.
        Quantum computers potentially exhibit an advantage in this regime for larger numbers of orbitals $N$.
        Then, quantum phase estimation can, in principle, obtain an energy of guaranteed accuracy that can be taken as a reference for standard coupled cluster models.
  \item \textbf{The intermediate correlation limit:} Here, static electron correlation can become important and multi-configurational approaches are more suitable. In this regime, traditional computation begins to face severe problems.
        Only coupled cluster models that can deal with a multi-configurational reference or that are of high excitation degree (including at least quadruple excitations) are applicable.
        However, the former are not unambiguously defined, whereas the latter are too costly for all but the smallest molecules.
        At the same time, generic multi-configurational models such as the complete active space self-consistent field (CASSCF) wavefunction and the MPS wavefunction optimized by DMRG are natural choices, but suffer from a lack of similarly accurate dynamical correlation methods to account for the fact that CASSCF and DMRG approaches are restricted to a few dozens of orbitals only.
        Nevertheless, only a small number of determinants (say, on the order of a dozen) will represent the state qualitatively well. Their superposition can be initialized efficiently as a guiding state on a quantum computer (provided that knowledge about these determinants can be obtained at comparatively little cost before a quantum computation).
        \item\label{it:strong correlation} \textbf{The strong correlation limit:} A large number of Slater determinants will be required in order to achieve a sufficiently high overlap with the target state.
        There are only a few examples known in ground state chemistry of this kind.
        A prominent class of examples are iron-sulfur clusters, where the overlap of the optimal Slater determinant with the ground state becomes very small \cite{leeEvaluatingEvidenceExponential2023}, while DMRG optimized MPS may still have large overlap. Such cases can be more routinely found in materials science, where materials are built from  many units with half-filled single-particle states, or in electronically excited states.
\end{enumerate}

Generally, to assess the \emph{quantum advantage} of quantum phase estimation, a better understanding of which systems have polynomially scaling conventional methods in practice on the one hand, and an understanding of which systems allow for good guiding states on the other, is required \cite{leeEvaluatingEvidenceExponential2023,fomichevInitialStatePreparation2023}.

\section{Fermionic formalism}\label{sec:setup}

We start by defining notation and recalling the formalism of fermionic quantum systems.
We consider a single-particle space $\cH$ of dimension $N$ which we may identify with $\CC^N$.
We refer to elements of the single-particle space $\cH$ as \emph{modes}; or equivalently in chemistry terminology as \emph{orbitals}.
The full Hilbert space is the Fock space
\begin{align*}
  \bigoplus_{n=0}^N \cH^{\wedge n} \cong (\CC^2)^{\ot N},
\end{align*}
which can be mapped to $N$ qubits.
Numbering an orthonormal basis of modes $j=1,2,\dots,N$, we associate creation and annihilation operators $a_j$ and $a_j^\dagger$, respectively, which satisfy the fermionic anticommutation relations
$$
  \{a_j,a_k\} = 0 = \{a_j^\dagger,a_k^\dagger\}\ ,\quad \{a_j,a_k^\dagger\} = \delta_{jk}\ .
$$
We define the vacuum state $\ket{\Omega}$ to be the unique mutual kernel of all the operators $a_j$, from which the creation operators $a_j^\dagger$ generate the full Hilbert space as a Fock space.
Given any normalized vector $x = (x_1,\dots,x_N) \in \CC^N$ we let
\begin{align*}
  a^\dagger(x) = \sum_{j=1}^N x_j a_j^\dagger, \qquad a(x) = \sum_{j=1}^N \overline{x_j} a_j
\end{align*}
be the operators that create and annihilate the mode $x$.

The operators
\begin{align*}
  \hat{n}(x) = a^\dagger(x)a(x), \qquad \hat{n} = \sum_{j=1}^N a_j^\dagger a_j
\end{align*}
are the number operator for mode $x$ and the total number operator.
If $\ket{\Psi}$ is an eigenvector of $\hat{n}(x)$ with eigenvalue 0 or 1 we say that mode $x$ is respectively unoccupied or occupied in the state $\ket{\Psi}$.
Given a state $\ket{\Psi}$, its \emph{covariance matrix} or \emph{one-body reduced density matrix} (1-RDM) is an operator on $\cH$ defined by
\begin{align*}
  \langle x, \covmat y \rangle =  \bra{\Psi} a^\dagger(x) a(y) \ket{\Psi}.
\end{align*}
After a choice of basis this gives an $N \times N$ matrix with entries $\covmat_{jk} = \bra{\Psi} a_j^\dagger a_k \ket{\Psi}$.

A \emph{Slater determinant} is a state of the form
\begin{align*}
  a^\dagger(x_n) \dots a^\dagger(x_1) \ket{\Omega},
\end{align*}
where $x_1, \dots, x_n$ are a collection of orthonormal modes in $\cH$.
Up to a phase, the Slater determinant only depends on the subspace $\mathcal X$ spanned by the $x_1, \dots, x_n$, and it is the state where the modes in $\mathcal X$ are occupied, and the modes in $\mathcal X^{\perp}$ are unoccupied.

We can also define the $2N$ Majorana operators
\begin{align*}
  c_{2j-1} = a_j + a_j^\dagger, \qquad c_{2j} = -i(a_j - a_j^\dagger).
\end{align*}
A \emph{Gaussian unitary} is a unitary $U$ acting on the Fock space such that
as an orthogonal transformation $O \in O(2N)$,
\begin{align*}
  U c_p U^\dagger = \sum_q O_{pq} c_q\
\end{align*}
for an orthogonal transformation $O \in O(2N)$.
A \emph{Gaussian state} is a state of the form $U\ket{\Omega}$, where $U$ is a Gaussian unitary operator. In particular, any Slater determinant is a Gaussian state.

A \emph{non-interacting} (or \emph{free}, or \emph{one-body}) Hamiltonian is a Hamiltonian of the form
\begin{align*}
  H_{\free} = \sum_{j,k} h_{jk} a_j^\dagger a_k,
\end{align*}
where $h$ is a Hermitian operator on $\cH$.
By choosing a basis for $\cH$ in which $h$ is diagonal, we can always transform this to the form
\begin{align*}
  H_{\free} = \sum_{j=1}^N \epsilon_j \tilde{a}_j^\dagger \tilde{a}_j
\end{align*}
with $\epsilon_1 \leq \epsilon_2 \leq \dots$.
The ground state space of a non-interacting Hamiltonian is spanned by a set of Slater determinants.
If $\epsilon_j < 0$ for $j \leq n$, and $\epsilon_j > 0$ for $j > n + k$, then the ground state space is spanned by the set of Slater determinants for which modes $1, \dots, n$ are occupied and $n + k, \dots, N$ are unoccupied.
In particular, if $\epsilon_j \neq 0$ for all $j$ then the ground state is unique.

In this work we are concerned with quantum impurity Hamiltonians, where a free Hamiltonian is perturbed by a non-negligible but spatially localized impurity term $H_{\imp}$. The term $H_{\imp}$ is localized in the sense that, written in terms of the fermionic operators $a_j$, $a_j^\dagger$, it only contains the modes $j\leq M$ for some constant $M$.

\begin{dfn}
  A \emph{quantum impurity Hamiltonian} is a Hamiltonian of the form
  \begin{align*}
    H = H_{\free} + H_{\imp},
  \end{align*}
  where $H_{\free}$ is a non-interacting Hamiltonian on $N$ modes, and $H_{\imp}$ is an interacting Hamiltonian on a subset of $M$ of the modes. 
  The \emph{quantum impurity problem} is the problem of computing the ground state energy of $H = H_{\free} + H_{\imp}$ to accuracy $\precision$, where $H_{\free}$ has bounded single-particle energies $\epsilon_j$.
\end{dfn}

In this computational problem, note that $H_{\imp}$ and $M$ are taken to be constant parameters: we seek an efficient solution in terms of the total system size $N$ and the precision $\precision$.

There are a few ways in which this problem can be simplified. Firstly, note that we can assume without loss of generality that the single-particle energies are non-negative. This follows by defining a new set of fermionic operators $b_j$ by
\begin{align*}
  b_j = \tilde{a}_j^\dagger \qquad & \text{ for } j = 1,\dots,n      \\
  b_j = \tilde{a}_j \qquad         & \text{ for } j = n+1,\dots,N\ ,
\end{align*}
where $n$ is maximal such that $\epsilon_n < 0$, as above. More generally, this corresponds to the transformation
\begin{align*}
  b(x) = \sum_{j=1}^n x_j a_j^\dagger + \sum_{j=n+1}^N \overline{x_j} a_j\ .
\end{align*}
Under this transformation, the free Hamiltonian takes the form (after a constant energy shift)
\begin{align*}
  H_{\free} = \sum_{j=1}^N \abs{\epsilon_j} b_j^\dagger b_j\ ,
\end{align*}
and the state vacuum state with respect to the $b_j$, $\ket{\Theta}$, is a ground state. Note that $\ket{\Theta}$ is a Slater determinant with respect to the $\tilde{a}_j$, via
\begin{align*}
  \ket{\Theta} = \prod_{j=1}^n \tilde{a}_j^\dagger \ket{\Omega}\ .
\end{align*}
For convenience, we can normalise the energy scale to assume without loss of generality that $\epsilon_j \in [0,1]$ for all $j$.

We denote by $\omega$ the energy gap of $H_{\free}$, which is equal to the smallest nonzero $\abs{\epsilon_j}$,
\begin{align*}
  \omega = \min_{j : \epsilon_j \neq 0} \abs{\epsilon_j}\ .
\end{align*}
In fact, for the purposes of estimating the ground state energy of $H$ to accuracy $\precision$, it turns out that we may assume $\omega > \precision / m = \Omega(\precision)$. This fact follows by truncating the low-energy modes of $H_{\free}$ and it is proved as Lemma 5 of \cite{bravyi2017complexity}.

In our analysis of the quantum impurity problem we will make particular use of the covariance matrix with respect to the ground state $\ket{\Psi}$ of $H$, whose entries are given by
\begin{align*}
  \covmat_{jk} = \bra{\Psi} b_j^\dagger b_k \ket{\Psi} \, .
\end{align*}

\section{Covariance matrix analysis for quantum impurity models}\label{sec:proof impurity}
In this section, $H = H_{\free} + H_{\imp}$ is a quantum impurity Hamiltonian on $N$ modes, with an impurity on $M$ modes and we use the notation of \cref{sec:impurity}. We assume that $H_{\free}$ has the form
\begin{align*}
  H_{\free} = \sum_{j=1}^N \epsilon_j a_j^\dagger a_j\ ,
\end{align*}
where $\epsilon_j \in [-1,1]$ for all $j$. As in \cref{sec:setup}, we may write
\begin{align*}
  H_{\free} = \sum_{j=1}^N |\epsilon_j| b_j^\dagger b_j
\end{align*}
for some new fermionic operators $b_j$. In this section we denote by $\covmat$ and $\covmat'$ the 1-RDMs of the ground state $\ket{\Psi}$ with respect to the $b_j$ and $a_j$ modes respectively, that is,
\begin{align*}
  \covmat_{jk} = \bra{\Psi} b_j^\dagger b_k \ket{\Psi} \qquad \covmat'_{jk} = \bra{\Psi} a_j^\dagger a_k \ket{\Psi}\ .
\end{align*}

In \cite{bravyi2017complexity} it is shown that for the ground state $\ket{\Psi}$ of $H$, the covariance matrix $\covmat$ has exponentially decaying eigenvalues.
To be precise, Theorem 2 of \cite{bravyi2017complexity} shows that there exists a constant $c$ such that there exists a ground state $\ket{\Psi}$ such that the eigenvalues $\sigma_1 \geq \sigma_2 \geq \dots$ of $\covmat$ are bounded as
\begin{align}\label{eq:exp decay 1rdm}
  \sigma_j \leq c\exp\mleft(- \frac{j}{14M\log(2\gap^{-1})}\mright) \, .
\end{align}
This bound can be used to show that $\ket{\Psi}$ can be approximated by a state that is a superposition over a limited number of Gaussian states. We revisit the argument of \cite{bravyi2017complexity} and show that (with respect to the original choice of fermionic operators $a_j$) we can also obtain a ground state approximation using Slater determinants. Note that the class of Gaussian quantum states is strictly larger than the class of Slater determinants; for certain models ground state approximations using Gaussian states can be much better than those by Slater determinants \cite{bravyi2019approximation}.

The key computational step is summarised in the following lemma. Informally, this guarantees that by assuming that modes with eigenvalues close to 1 in the 1-RDM are filled, and assuming that those with eigenvalues close to 0 are unfilled, one can obtain a reasonable approximation to the true state.

\begin{lem}\label{lem:excitation projection}
  Let $\covmat$ be the 1-RDM of a state $\ket{\Psi}$, and let $x_1,x_2,\dots,x_N$ be an orthonormal basis for $\CC^N$. Given disjoint subsets $I^+,I^- \subseteq [n]$, define the projectors
  \begin{align*}
    \Pi^- = \prod_{j \in I^-} a(x_j)^\dagger a(x_j)\qquad \Pi^+ = \prod_{j \in I^+} a(x_j) a(x_j)^\dagger\ .
  \end{align*}
  Then there exists a state $\ket{\tilde\Psi}$ in the image of both $\Pi^-$ and $\Pi^+$ such that
  \begin{align*}
    |\braket{\Psi}{\tilde\Psi} | \geq 1 - \delta\ ,
  \end{align*}
  where
  \begin{align*}
    \delta \leq \sum_{j \in I^-} \sqrt{1-\covmat_{jj}} + \sum_{j \in I^+} \sqrt{\covmat_{jj}} \ ,
  \end{align*}
  where $\covmat_{jk}$ are the elements of $\covmat$ with respect to the basis $\{x_j\}$.
\end{lem}

\begin{proof}
  We can directly compute
  \begin{align*}
    \bra{\Psi}  \Pi^+ \Pi^- \ket{\Psi} & = 1 - \bra{\Psi} I - \Pi^+ \Pi^- \ket{\Psi}                          \\
                                       & \geq 1 - \| (I - \Pi) \ket{\Psi} \|                                  \\
                                       & \geq 1 - \| (I - \Pi^-) \ket{\Psi} \| - \|(I - \Pi^+)\ket{\Psi}\|\ .
  \end{align*}
  The first norm may be bounded by
  \begin{align*}
    \|I - \Pi^-\ket{\Psi} \| & \leq \sum_{j \in I^-} \| I - a(x_j)^\dagger a(x_j) \ket{\Psi} \| \\
                             & = \sum_{j\in I^-} \sqrt{1-\covmat_{jj}}\ ,
  \end{align*}
  and the second may be bounded by
  \begin{align*}
    \|I - \Pi^+\ket{\Psi} \| & \leq \sum_{j \in I^+} \| I - a(x_j) a(x_j)^\dagger \ket{\Psi} \| \\
                             & = \sum_{j \in I^+} \sqrt{\covmat_{jj}}\ .
  \end{align*}
  We now let
  \begin{align}\label{eq:projected state}
    \ket{\tilde\Psi} = \frac{\Pi^+ \Pi^- \ket{\Psi}}{\|\Pi^+ \Pi^- \ket{\Psi}\|}\ ,
  \end{align}
  which satisfies the theorem by the above calculation.
\end{proof}

We now restate \cref{thm:Slater det impurity}, with some additional details.

\begin{thm}
  Let $\gap > 0$ be the ground state energy gap of $H_{\free}$, and let $\eps > 0$. Then for
  \begin{align*}
    K = \bigO\mleft(\log(\gap^{-1}) \mleft(\log(\eps^{-1}) + \log \log(\gap^{-1}) \mright) \mright)
  \end{align*}
  there exists a Slater determinant $\ket{\Theta}$ on $N - K$ modes and an arbitrary state $\ket{\Phi}$ on $K$ modes such that the state
  \begin{align*}
    \ket{\tilde\Psi} = \ket{\Phi} \wedge \ket{\Theta}
  \end{align*}
  has ground state overlap $\abs{\braket{\tilde\Psi}{\Psi}} \geq 1 - \eps$.
  Moreover, we may choose the Slater determinant $\ket{\Theta}$ to be defined on a set of modes that commute with $H_{\imp}$.
\end{thm}

\begin{proof}
  The full single-particle space is $\CC^N$.
  Let $\mathcal J_+$ and $\mathcal J_-$ be the subspaces of non-negative and negative energy modes (of dimensions $N - n$ and $n$ respectively) and let $\mathcal M$ be the subspace of modes on which $H_{\imp}$ acts (of dimension $M$), and let $\mathcal L$ be the orthogonal complement of $\mathcal M$. By definition, $H_{\imp}$ commutes with $a(x)$ for any $x \in \mathcal L$.
  We then define
  \begin{align*}
    \mathcal L_+ = \mathcal J_+ \cap \mathcal L, \qquad \mathcal L_- = \mathcal J_- \cap \mathcal L.
  \end{align*}
  By dimension counting,
  \begin{align*}
    \dim(\mathcal L_+) \geq N - M - n \qquad \dim(\mathcal L_-) \geq n - M.
  \end{align*}
  We define fermionic operators $b_j$ by $b_j = a_j^\dagger$ for $j \leq n$ and $b_j = a_j$ for $j > n$, so the free Hamiltonian is given by
  \begin{align*}
    H_{\free} = \sum_{j=1}^N \abs{\epsilon_j} b_j^\dagger b_j
  \end{align*}
  after subtracting a constant term.
  Since $H_{\free}$ is assumed to be normalized, the single particle energies $|\epsilon_j|$ now lie in $[0,1]$.
  We let $\covmat$ be the 1-RDM for the ground state $\ket{\Psi}$ with respect to the fermionic operators $b_j$, with eigenvalues $\sigma_1 \geq \sigma_2 \geq \dots$.
  Let $\Lambda^{\pm}$ be the projection onto $\mathcal L_{\pm}$ and let $\lambda_1^{\pm} \geq \lambda_2^{\pm} \geq \dots$ be the eigenvalues of $\Lambda^{\pm} \covmat  \Lambda^{\pm}$.
  By the Cauchy interlacing theorem and \cref{eq:exp decay 1rdm} we have
  \begin{align*}
    0 \leq \lambda_j^{\pm} \leq \sigma_j \leq c\exp\mleft(- \frac{j}{14M\log(2\omega^{-1})}\mright).
  \end{align*}
  Let $\covmat'$ be the 1-RDM with respect to the original fermionic operators $a_j$, so that
  \begin{align*}
    \covmat'_{jk} & = \delta_{jk} - \covmat_{kj} & \text{for } j,k \leq n, \\
    \covmat'_{jk} & = \covmat_{jk}               & \text{for } j,k > n.
  \end{align*}
  In particular, this means that $\Lambda^+ \covmat' \Lambda^+$ has eigenvalues $\lambda_j^+$ and $\Lambda^- \covmat' \Lambda^-$ has eigenvalues $1 - \lambda_j^-$.

  Let $\{v_j^+\} \subseteq \mathcal{L}^+$ and $\{v_j^-\} \subseteq \mathcal{L}^-$ be eigenbases of $\Lambda^+ \covmat  \Lambda^+$ and $\Lambda^- \covmat  \Lambda^-$ respectively, corresponding to eigenvalues of descending size. The union of these bases can be extended to an orthonormal basis $\mathcal{B}$ for $\CC^N$. Let $I^-$ index all but the first $K/2-M$ elements of $\{x_j^-\}$, and let $I^+$ index all but the first $K/2-M$ elements of $\{x_j^+\}$. Note that either of these sets may be empty, but $|I^- \cup I^+ | \geq N - K$.

  Applying \cref{lem:excitation projection} with $I^+$ and $I^-$ as above, we obtain a state $\ket{\tilde\Psi}$ as in \cref{eq:projected state} which has ground state overlap bounded by

  \begin{align*}
    |\braket{\Psi}{\tilde\Psi} | \geq 1 - \sum_{j \geq K/2 - M} \Big(\sqrt{\lambda_j^{+}} + \sqrt{\lambda_j^-}\Big)\ .
  \end{align*}

  By construction, the state $\ket{\tilde\Psi}$ is of the form
  \begin{align*}
    \ket{\tilde\Psi} = \ket{\Phi} \wedge \ket{\Theta}\ ,
  \end{align*}
  where $\ket{\Theta}$ is the state with modes $v_j^{+}$ not filled and modes $v_j^{-}$ filled for $j \geq K/2 - M$, and where $\ket{\Phi}$ is a state on at most $K$ modes. By choosing $K$ as in the theorem statement, the ground state overlap can be lower bounded by $1-\delta$, and hence $\ket{\tilde\Psi}$ fulfils all the desired requirements.
\end{proof}

In the case where $H_{\free}$ has a constant spectral gap $\omega$, the exponential decay of the 1-RDM spectrum places a constant upper bound on the number of excitations for a good approximation to the ground state $\ket{\Psi}$. This leads to a space of polynomial dimension which may be readily searched via quantum phase estimation. If the matrix $\covmat$ itself is known --- but with no assumptions on the spectral gap $\omega$ --- a similar argument can be applied. Although one cannot place a constant bound on the total number of excitations as before, knowing the eigenvectors of $\covmat$ allows one to predict \emph{which} modes are likely to be excited. Following the same approach as Corollary 2 in \cite{bravyi2017complexity}, this again leads to a search space of polynomial dimension. Based on the above discussion, we now restate and prove \cref{thm:complexity impurity informal}.

\begin{thm}\label{thm:impurity complexity}
  Consider a quantum impurity problem with $M = \bigO(1)$.
  Suppose that either $\gap = \bigO(1)$, or we are given the covariance matrix $\covmat$ for a ground state satisfying \cref{eq:exp decay 1rdm}.
  Then the quantum impurity problem can be solved by a quantum computer using $\poly(N, \precision^{-1})$ gates.
\end{thm}

It is not necessary that $\covmat$ is given precisely; for the proof below it is sufficient merely to have knowledge of an upper bound $\tilde{\covmat} \geq \covmat$ such that the spectrum of $\tilde{\covmat}$ decays exponentially as in \cref{eq:exp decay 1rdm}.

\begin{proof}
  We may choose fermionic operators such that
  \begin{align*}
    H_{\free} = \sum_j \epsilon_j b_j^\dagger b_j
  \end{align*}
  with $0 \leq \epsilon_j \leq 1$.
  We represent the Hamiltonian on $N$ qubits using the Jordan-Wigner transformation.
  If $\gap = \bigO(1)$, by \cref{thm:Slater det impurity}, for constant $\eps = \frac12$, for \emph{constant} $K$ there exists a ground state $\ket{\Psi}$ that has overlap $\braket{\Psi}{\tilde\Psi} \geq \frac12$ for a state $\ket{\tilde\Psi}$ that is a superposition of Slater determinants with at most $K$ excitations.
  We now consider the subspace $V$ of the full Hilbert space, which is spanned by all states with at most $K$ excitations.
  The dimension of this space is
  \begin{align}\label{eq:dimension bound}
    \dim(V) = \sum_{k=0}^K {N \choose k}.
  \end{align}
  For $\nicefrac{N}{2} \geq K$
  \begin{align*}
    \sum_{k=0}^K {N \choose k} \leq K {N \choose K} \leq K N^K = \poly(N).
  \end{align*}
  For small $\nicefrac{N}{2} < K$, \cref{eq:dimension bound} is bounded by $2^N \leq 2^{2K}$, which is constant.
  If we let $\tau$ denote the maximally mixed state on $V$, then
  \begin{align*}
    \bra{\Psi} \tau \ket{\Psi} \geq \frac{\abs{\braket{\Psi}{\tilde\Psi}}^2}{\dim(V)},
  \end{align*}
  which has at least inverse polynomial magnitude.
  The mixed state $\tau$ can be efficiently prepared on a quantum computer.
  In the qubit representation it corresponds to the uniform mixture of standard basis states $\ket{x_1\dots x_N}$, $x_1\dots x_N \in \{0,1\}^N$ with Hamming weight $x_1 + \dots + x_N \leq K$.
  One can prepare this state by uniformly sampling from bit strings $x_1 \dots x_n$ with Hamming weight at most $K$, and then prepare $\ket{x_1 \dots x_n}$.
  Applying quantum phase estimation, using (approximate) time evolution along $H$ and initial state $\tau$ now gives an efficient algorithm for computing the ground state energy to precision $\precision$ using $\poly(N, \precision^{-1})$ gates.

  The case where we are given $\covmat$ (but no assumption on $\gap$) proceeds similarly; it suffices to find an efficiently preparable state with at least inverse polynomial overlap.

  First we note that, as discussed in \cref{sec:setup}, we may without loss of generality restrict ourselves to the case where $\omega = \Omega(\precision)$. Let $x_1,x_2,\dots \in \CC^n$ be the orthonormal eigenvectors of the 1-RDM $\covmat$ corresponding to eigenvalues $\sigma_1\geq \sigma_2\geq\dots$ respectively. Now take $Q = \lceil 14M \log (2\omega^{-1})\rceil$, and define sets
  \begin{align*}
    I_1    & = \{1,2,\dots,Q\},   \\
    I_2    & = \{Q+1,\dots,2Q\},  \\
    I_3    & = \{2Q+1,\dots,3Q\}, \\
    \vdots &
  \end{align*}
  For each $s \in \NN$, define the partial number operator
  $$
    N_s = \sum_{j \in I_s} b(x_j)^\dagger b(x_j)\ ,
  $$
  which counts how many of the modes corresponding to the subset $I_s$ of eigenvectors of $\covmat$ are excited. Note that the $N_s$ mutually commute. We can upper bound the expected value of these observables by
  \begin{align*}
    \bra{\Psi} N_s \ket{\Psi} & = \sum_{j \in I_s} \sigma_j             \\
                              & \leq c_0 M \log(2\omega^{-1}) e^{-s}\ ,
  \end{align*}
  where $c_0$ is some universal constant, using \eqref{eq:exp decay 1rdm}. For each $s$, let $n_s$ be a random variable corresponding to the measurement distribution of the observable $N_s$ induced by the state $\ket{\Psi}$, and let $R_s = c_0 M\log(2\omega^{-1}) e^{-s/2}$. Then Markov's inequality implies
  $$
    \prob[n_s \geq R_s] \leq e^{-s/2}\ ,
  $$
  and applying a union bound we see that, for any $S \in \NN$,
  $$
    \prob[n_s \leq R_s \text{ for all } s \geq S] \geq 1 - \sum_{s\geq S} e^{-s/2}\ .
  $$
  Now choose $S_0 \in \NN$ such that
  $$
    \sum_{s \geq S_0} e^{-s/2} \leq \frac{1}{2}
  $$
  and
  $$
    \frac{R_s}{Q} < \frac{1}{2} \text{ for all } s \geq S_0\ .
  $$
  Note that such an $S_0$ can be chosen as a universal constant, independent of $M$, $\omega$, and $N$. Then
  $$
    \prob[n_s \leq R_s \text{ for all } s \geq S_0] \geq \frac{1}{2}\ .
  $$
  Equivalently, letting $P$ denote the projection onto the subspace spanned by Fock basis states $\ket{\Phi}$ (with respect to the modes $b(x_j)$) satisfying $N_s \ket{\Phi}\leq R_s$ for all $s > S_0$,
  $$
    \bra{\Psi} P \ket{\Psi} \geq \frac{1}{2}\ .
  $$
  In particular this implies that the maximally mixed state on $P$, $\tau^\prime$, has squared ground state overlap
  $$
    \bra{\Psi} \tau^\prime \ket{\Psi} \geq \frac{1}{2\Tr[P]}\ .
  $$
  By an identical combinatorial argument to the one presented in Corollary 2 of \cite{bravyi2017complexity}, $\Tr[P]$ can be bounded by $e^{O(m\log(\omega^{-1}))} = \poly(\precision^{-1})$, completing the proof.
\end{proof}

In fact even in the absence of a spectral gap or knowledge of the 1-RDM as in \cref{thm:impurity complexity}, one may still obtain a quasipolynomial quantum speedup for the impurity problem. In particular, whereas the classical algorithm of \cite{bravyi2017complexity} requires time $\poly(N)\exp\Bigl( \bigO(\log(\precision^{-1})^3)\Bigr)$, a naive quantum adaptation of this approach can reduce the time complexity to $\poly(N)\exp\Bigl( \bigO(\log(\precision^{-1})^2)\Bigr)$. This speedup arises from the step analogous to \cref{lem:excitation projection}, in which the search space for candidate ground states is reduced to a limited set of active excitations. In the classical case it is necessary to choose $\delta = \bigO(\precision)$ to ensure that $\ket{\tilde\Psi}$ approximates the ground state energy sufficiently well, however in the quantum case $\delta$ can be taken as constant for the phase estimation step. In the algorithm of \cite{bravyi2017complexity}, this leads to reduction of the search space dimension by a factor of $\log(\precision)$. Although this is not sufficient to show BQP containment of the quantum impurity problem, it provides intuition that a polynomial time algorithm for the fully general case may be attainable through more detailed analysis.

\section{Numerical methods}\label{sec:numerical_methods}
Geometries of all systems were optimized with RI-DFT \cite{eickhorn1995RIDFT,eickhorn1997RIDFT} and a Perdew, Burke, and Ernzerhof exchange-correlation functional (PBE) \cite{perdew1996pbe}, D3 dispersion correction \cite{grimme2010d3} and Becke-Johnson damping \cite{grimme2011BJ_correction}
using a def2-SVP basis set \cite{weigend2005def2svp}. Structures were optimized with TURBOMOLE \cite{Ahlrichs1989,Turbomole702}.

For the (oligo)tryptophan systems, HF molecular orbitals were obtained for the bootstrap embedding and for the sequence of tryptophan side chains. As we were interested in the valence shell for the exploration of multiconfigurational character, the molecular orbitals were first obtained using the def2-SVP basis set and then used to construct intrinsic bond orbitals \cite{Knizia2013, Senjean2021} (IBOs) separately for occupied and virtual space.
For the Huzinaga embedding \cite{hegely2016exact, Chulhai2018} in the tryptophan and ruthenium system, Kohn-Sham orbitals calculated with PBE were used instead. IBOs were again constructed separately for the occupied and virtual space. These calculations were run with Serenity \cite{Serenity2018, Niemeyer2022}. For the bootstrap embedding, Hartree-Fock molecular orbitals were obtained using the def2-SVP basis set. Following the Hartree-Fock calculation, intrinsic atomic orbitals \cite{Knizia2013} were constructed and used to represent the 1-RDM. From there, the procedure described in Ref.~\cite{ye2021bootstrap} was followed to obtain the fragment and the entangled bath orbitals. These calculations were run with the PySCF program \cite{sun2018pyscf,sun2020pyscf,sun2015libcint}.

Orbitals selected for the active space calculations in the case of the tryptophan sequence comprised the entire $\pi$ system of the tryptophan side chain, as well as the $\sigma$ and $\sigma^*$ orbitals of the C$-$H bond in the five-membered ring that forms a weak CH-$\pi$ hydrogen bond. For the embedding examples, single-orbital entropies, mutual information and orbital ordering were obtained with autoCAS \cite{stein2016automated,autoCAS230} and the active spaces were constructed by choosing $K$ orbitals corresponding to the largest values.

Next, DMRG calculations with bond dimension $D=1024$ were performed to obtain an approximation of the target ground state of the system as an MPS.
Slater determinants with the leading contributions were constructed from the MPS using sampling-reconstruction of the complete active space (SR-CAS) \cite{boguslawski2011srcas} and used to construct the sum-of-Slater states. MPS wavefunctions of lower bond dimensions were obtained by truncating the bond dimension using the singular value decomposition. DMRG calculations were run with the QCMaquis software package \cite{keller2015qcmaquis}.

%\bibliographystyle{unsrtnat}
%\bibliography{bibliography}

\end{document}